\documentclass[12pt]{article}
\usepackage[english]{babel}    
\usepackage[ansinew]{inputenc} 
\usepackage[T1]{fontenc}       
\usepackage{lmodern,microtype} 
\usepackage{amsmath,amsthm,amssymb,amsfonts} 
\usepackage{dsfont,mathrsfs,ushort} 

\usepackage{titlesec,titling}  
\usepackage[nohead]{geometry}  
\usepackage{setspace,caption}  
\usepackage{enumitem,booktabs} 
\usepackage[comma,longnamesfirst]{natbib} 
\usepackage{pstricks,pst-all}   
\usepackage{tikz}
\usepackage{pst-plot,pst-node,pst-3dplot}

\usepackage{pgfplots}
\usepackage{graphicx}
\usepackage{caption}
\usepackage{subcaption}

\usepackage{hyperref,breakurl}
\hypersetup{breaklinks=true}
\usepackage{longtable}
\usepackage[flushleft]{threeparttable}

\usepackage{algorithmic}
\setenumerate{label=\small(\roman*)}
\hypersetup{colorlinks=true,pdfnewwindow=true,pdfstartview=FitH,%
	pdftitle="PPPP",pdfauthor="Victor Aguiar Roy Allen and Nail Kashaev",%
	urlcolor=blue!90!red!45!black,citecolor=blue!90!red!45!black,linkcolor=red!90!black}
\DeclareCaptionFont{fancy}{\bfseries\sffamily}
\allowdisplaybreaks

\providecommand{\psreset}{\psset{%
		linewidth=0.3pt,linestyle=solid,linecolor=black,
		dotsize=2.5pt,dotsep=2.5pt,arrowsize=4pt,
		fillstyle=none,fillcolor=white,
		showpoints=false,arrows=-,linearc=0,framearc=0,
		hatchsep=2pt,hatchwidth=0.2pt,nodesep=4pt,opacity=1}
	\psset{gridcolor=black!60, subgridcolor=black!30}
}

\psreset

\usepackage{graphicx}
\graphicspath{ {figures/} }

\titleformat{\section}[block]{\centering\large\bfseries\sffamily}{\thesection.}{0.5em}{}
\titleformat{\subsection}[block]{\flushleft\bfseries}{\thesubsection.}{0.5em}{}
\titleformat{\subsection}[block]{\flushleft\bfseries\sffamily}{\thesubsection.}{0.5em}{}
\titleformat{\subsubsection}[runin]{\normalsize\itshape}{\bfseries\upshape\sffamily\thesubsubsection.}{0.5em}{}[.--\:]
\renewcommand{\thesubsubsection}{\arabic{section}.\arabic{subsection}.\alph{subsubsection}}
\titlespacing{\section}{0ex}{10ex}{5ex}
\titlespacing{\subsection}{0in}{6ex}{3ex}
\titlespacing{\subsubsection}{0mm}{2ex}{0.5em}

\providecommand{\abstitle}[1]{{\par\vspace*{2ex}\small\bfseries\sffamily #1}\hspace*{1ex}}
\renewenvironment{abstract}%
{\begin{center}\begin{minipage}{0.8\linewidth}%
			\setlength{\parindent}{0.0em}\abstitle{Abstract}\small}%
		{\end{minipage}\end{center}\vfill\clearpage}

\newtheorem{proposition}{Proposition}[section]

\DeclareMathOperator*{\hull}{co}

\providecommand{\Real}{{\mathds{R}}}
\providecommand{\Natural}{{\mathds{N}}}

\providecommand{\tr}{^{\prime}}
\providecommand{\as}{\ensuremath{\mathrm{a.s.}}}
\providecommand{\rand}[1]{\mathbf{#1}}
\providecommand{\rands}[1]{\boldsymbol{#1}}

\newcommand{\norm}[1]{\left\lVert#1\right\rVert}
\newcommand{\Prob}[1]{\mathds{P}\left(#1\right)}
\newcommand{\Exp}[1]{\mathds{E}\left[#1\right]}
\newcommand{\abs}[1]{\left\lvert#1\right\rvert}

\newcommand{\x}{x}\newcommand{\X}{X} 
\newcommand{\y}{y}\newcommand{\Y}{Y} 
\newcommand{\e}{e}\newcommand{\E}{E} 
\newcommand{\p}{p} \renewcommand{\P}{P} 
\newcommand{\yr}{\y_{-z}} 
\newcommand{\yu}{\y_{z}} 
\newcommand{\pu}{\p_{z}}   

  \theoremstyle{remark}
  \newtheorem{rem}{\protect\remarkname}
  \theoremstyle{plain}
  \newtheorem{lem}{\protect\lemmaname}
  \theoremstyle{definition}
  \newtheorem{defn}{\protect\definitionname}
\theoremstyle{plain}
\newtheorem{thm}{\protect\theoremname}

  \theoremstyle{plain}
  
 \theoremstyle{definition}
  \newtheorem{example}{\protect\examplename}
  \theoremstyle{plain}
  \newtheorem{assumption}{\protect\assumptionname}
\newtheorem{prop}{Proposition}

  \providecommand{\assumptionname}{Assumption}

  \providecommand{\definitionname}{Definition}
  \providecommand{\lemmaname}{Lemma}
  
  \providecommand{\remarkname}{Remark}
\providecommand{\corollaryname}{Corollary}
\providecommand{\theoremname}{Theorem}
\providecommand{\examplename}{Example}

\usepackage{xcolor}

\makeatother

\begin{document}
\title{Prices, Profits, Proxies, and Production \thanks{The ``\textcircled{r}'' symbol indicates that the authors' names are in certified random order, as described by \citet{ray2018certified}. An earlier version of this paper was circulated as ``Prices, Profits, and Production: Identification and Counterfactuals.''}}
\author{ 
	Victor H. Aguiar \textcircled{r}
	Nail Kashaev \textcircled{r}
	Roy Allen\thanks{Aguiar: Department of Economics, University of Western Ontario; \href{mailto:vaguiar@uwo.ca}{vaguiar@uwo.ca}. Kashaev: Department of Economics, University of Western Ontario; \href{mailto:nkashaev@uwo.ca}{nkashaev@uwo.ca}. Allen: Department of Economics, University of Western Ontario; \href{mailto:vaguiar@uwo.ca}{rallen46@uwo.ca}.}
	}
\date{First version: October 10, 2018\\ 
This version: June, 2022}
\maketitle
\begin{abstract}
This paper studies nonparametric identification and counterfactual bounds for heterogeneous firms that can be ranked in terms of productivity. Our approach works when quantities and prices are latent, rendering standard approaches inapplicable. Instead, we require observation of profits or other optimizing-values such as costs or revenues, and either prices or price proxies of flexibly chosen variables. We extend classical duality results for price-taking firms to a setup with discrete heterogeneity, endogeneity, and limited variation in possibly latent prices. Finally, we show that convergence results for nonparametric estimators may be directly converted to convergence results for production sets.
\bigskip


\end{abstract}

\section*{Introduction}

This paper studies nonparametric identification of production sets and counterfactual bounds for firms, allowing multiple inputs and outputs, in an environment where both quantities and prices can be latent. 
We assume an analyst has data on the values of an optimization problem, such as profits, costs, or revenues, as well as prices or \textit{price proxies}.
\par
Identifying heterogeneous production sets is challenging in situations where the observability of some outputs/inputs or prices is problematic. For instance, in the housing market output quantities and output prices cannot be directly observed because houses provide different services that are hard to measure. However, housing values that can serve as price proxies may be observed \citep{epple2010new}. Other industries, such as health and banking, suffer from similar issues with unobservable inputs or outputs.\footnote{In the health industry,  it is difficult to measure inputs such as drugs since they vary widely in their physical characteristics. However, prices and total costs may be observable \citep{bilodeau2000hospital}. In the banking industry, outputs such as business loans and consumers loans are difficult to measure because a loan is a financial service that entails many unobservable goods and services. However, the price of a loan is observed as well as profits in some settings \citep{berger1993bank}. } The latency of quantities makes standard approaches to estimate production functions not directly applicable. In addition, the latency of prices makes classical approaches using duality theory impossible to apply as well.  In contrast, we require observability of values and prices or price proxies. While these variables are not always observed, they are available in many existing data sets.\footnote{See \citet{epple2010new,combes2017production}, and \citet{albouy2018housing} in the context of housing; \citet{burke2019sell} in the context of agriculture; \citet{nerlove63} and \citet{costelectricityKira2007} in the context of electricity generation; \citet{roberts1996output, foster2008reallocation}, and \citet{doraszelski2013r} in the context of manufacturing.} 
\par 
In order to obtain identification of firm-specific production possibility sets we exploit variation in prices or price proxies across markets and variation of optimization values across firms. Our framework extends classical duality theory by allowing (i) rich forms of complementarity and substitutability between outputs and inputs with discrete heterogeneity across firms, (ii) endogeneity between prices and productivity due to simultaneity and market entry decisions, and (iii) omitted prices of flexibly chosen variables.
Classical duality theory focuses on either a nonstochastic or representative agent framework in which all prices are observed. Important contributions include \citet{shephard1953,fussmcfadden2014production}, and \citet{diewert1982duality} among many others.  
\par 
We assume that firms can be ranked in terms of productivity that can take finitely many values. This assumption is key to unpack heterogeneity in multiple output/input production sets across firms from data such as prices or price proxies and scalar values of an optimization problem. We formalize this by assuming that a firm with higher productivity has access to all the production possibilities of a less productive firm, and more. Our framework covers Hicks-neutral heterogeneity in productivity as a special case. 
\par
Our approach exploits the rich shape constraints in our environment for identification and counterfactual analysis. Leveraging that firms can be ranked according to discrete productivity, we present a new method to identify the structural value function (e.g. profit function). This technique works with bounded measurement error, but allows rich forms of selection into market. We require a weak monotone presence assumption, so that if a firm is present in some market with certain observables, then each more productive firm must be present in some market with the same observables. This handles certain monotone selection rules, e.g. only firms that can make nonnegative profits enter, but is much more general.

We next tackle the important possibility that not all prices are observed. Instead, we use \textit{price proxies}, which are unknown functions of the missing prices. As one example, we show that aggregate market-level quantities can serve as price proxies. We leverage homogeneity of the value function to recover these unknown functions. This technique is new, and is applicable to other settings with homogeneity of a structural function, and is therefore of independent interest.

Once the structural value function is identified, we turn to recoverability of the production sets. Here we leverage the classic insight that the value function serves as the support function of the production set. This allows us to characterize the most that can be said about heterogeneous production sets, even when price variation is limited. Building on this, we present a general framework for counterfactual questions such as sharp bounds on quantities or profits at a new price. Importantly, these bounds hold for each level of productivity, and thus characterize features of the distribution of firm behavior.
\par 
As mentioned previously, relative to classic work on duality we make several contributions by incorporating heterogeneity, endogeneity due to selection, and potential lack of prices.\footnote{Outside of the firm problem, duality has been used in the presence of heterogeneity in discrete choice \citep{mcfadden1981econometric}, matching models \citep{galichon2015cupid}, hedonic models \citep{chernozhukov2017single}, dynamic discrete choice \citep{chiong2016duality}, and the additively separable framework of \citet{allen2019identification}.} Even when prices are observed but contain limited variation, we contribute by providing new results using structural value functions to recover sets and conduct counterfactual analysis. This builds on \citet{farrell1957measurement} and \citet{afriat1972efficiency}, who study efficiency measurement and conditions under which producer datasets are consistent with the hypothesis of optimization. Relatedly, \citet{hanoch1972testing} focuses on finite deterministic datasets of individual firms' profits or costs, and prices. \citet{hanoch1972testing} does not study identification of the production set or the profit function, but focuses on providing necessary and sufficient conditions under which an observed production function is consistent with profit maximization or cost minimization.\footnote{\citet{CHERCHYE2016100} studies the identification of profits and production sets with a finite deterministic dataset on prices and quantities.}   Another paper studying limited price variation is \citet{varian1984nonparametric}, which works with quantities and prices and does not study unobservable heterogeneity.\footnote{See also \citet{cherchye2014non} and \citet{cherchyedemuynck2018}. \citet{cherchyedemuynck2018} differs from this paper because they assume observed input quantities in the context of cost minimization.} While observation of prices and quantities implies observation of profits, the reverse is not true.
\par
This paper contributes to the recent literature on identification and estimation of multi-output production with unobservable heterogeneity (e.g., \citealp{cunha2010estimating}, \citealp{de2016prices}, and \citealp{grieco2016productivity}). We differ since we do not observe quantities and we do not impose separability or parametric restrictions on the shape of production sets.
Because we allow production of multiple outputs in flexible ways, use cross-sectional variation, and do not observe quantities, we also differ from an important recent literature studying single output production in dynamic panel settings using quantities data, including \citet{griliches1995production,olley1996dynamics,levinsohn2003estimating,ackerberg2015identification}, and \citet{gandhi2017identification}.\footnote{As noted in \citet{ackerberg2015identification}, some output and input data often come in the form of sales and expenditures that need to be transformed into quantities. We work directly with total values (e.g. profits, total costs, or revenues).}
\par
We also contribute to the literature studying recoverability of sets. We build on the tight relationship between the structural value function and the production possibility sets of firms, by providing an equality relating estimation error of value functions and estimation error of production possibility sets. This result allows one to adapt consistency results for any nonparametric estimators of the value function for the purpose of set estimation. The result is related to a classical result in convex analysis linking the distance of support functions with the distance of the corresponding sets, which has been exploited previously in the literature on partial identification.\footnote{See, for instance, \citet{beresteanu2008asymptotic,beresteanu2011sharp,kaido2014asymptotically,kaido2016dual}, and \citet{kaido2019confidence}.} 
We cannot apply the classical result since it would require seeing negative prices.
\par
The rest of this paper proceeds as follows. In Section~\ref{sec:model}, we present a model of heterogeneous production in which firms are rankable in terms of productivity. Section~\ref{sec: restricted prof func} shows how to identify the structural value function. In Section~\ref{sec:homog}, we extend our methodology to environments where one observes proxies that determine unobservable prices. Our main identification result for production possibility sets is in Section~\ref{sec: identification general}. Section~\ref{sec:sharpbounds} provides a general framework to conduct sharp counterfactual analysis in production environments. In Section~\ref{sec:estimation}, we show duality between estimation error in value functions and production sets. We conclude in Section~\ref{sec:conclusion}. All proofs can be found in Appendix~\ref{app: proofs}. An estimator of the restricted profit function and an illustrative application are in Appendices~\ref{app: estimation and simulation} and~\ref{app: application}. The Online Appendix contains extensions, simulations, and additional results.

\section{Setup}\label{sec:model}
This paper studies recoverability of the technology of heterogeneous firms given data on the value function of their maximization problems, as well as data on prices or price proxies that alter the maximization problems. 

The technology of heterogeneous firms is described by a correspondence $\Y:\E\rightrightarrows \Real^{d_{\y}}$. Each set $\Y(\e)$ describes the possible input/output (or ``netput'') vectors that are feasible for a firm of type $\e$. The variable $\e$ captures unobservable heterogeneity in productivity. Negative components of $\Y(\e)$ correspond to net demands by the firm and positive components correspond to net supply. This formulation allows us to treat single output and multi-output firms in a common framework.\footnote{An alternative approach is to use transformation functions. See \citet{grieco2016productivity} for a recent application.} We require the following conditions.
\begin{defn} \label{defn:production correspondence}
A correspondence $\Y:\E\rightrightarrows \Real^{d_{\y}}$ is a production correspondence if, for every $\e\in\E$,
\begin{enumerate}
\item $\Y(\e)$ is closed and convex;
\item $\Y(\e)$ satisfies \emph{free disposal}: if $\y$ is in $\Y(\e)$, then any $\y^*$ such that $\y^*_j\leq \y_j$ for all $j\in \{1,\cdots,d_\y\}$ is also in $\Y(\e)$;
\item $\Y(\e)$ satisfies \emph{the recession cone property}: if $\{\y^m\}$ is a sequence of points in $\Y(\e)$ satisfying $\norm{\y^m} \to \infty$ as $m\to\infty$, then accumulation points of the set $\{ \y^m / \norm{\y^m} \}_{m = 1}^{\infty}$ lie in the negative orthant of $\Real^{d_\y}$.
\end{enumerate}
\end{defn}
These conditions rule out infinite profits and ensure that the maximization problems we consider have a solution.\footnote{See \citet{kreps2012}, p. 199 for more details.}

We study the general \textit{restricted profit maximization problem}
\[
\pi_r(\yr,\p_z,\e)=\max_{\yu : (\yr, \yu) \in \Y(\e)} \p_z \tr \yu\,,
\]
where $\yr$ is a vector of restricted or fixed variables, $\yu$ denotes the variables of choice, and $p_z$ is a vector of prices of $\yu$. The variable of choice $\yu$ is constrained to belong to the convex set $\Y_r(\yr,\e)$ defined as

\[
    \Y_r(\yr,\e)=\left\{\yu\in\Real^{d_{\yu}}\::\:(\yr,\yu)\in\Y(\e)\right\}\,.
\]
We refer to $\Y_r(\yr,\cdot)$ as the restricted production correspondence.\footnote{More formally, it is only a multi-valued mapping because it can be empty for certain combinations of $\yr$ and $\e$. We note that the results in this paper do not need the full strength of $Y(\cdot)$ being a production correspondence. Instead, we require that the set $\Y_r(\yr,\e)$ be closed and convex, satisfy free disposal, and satisfy the recession cone property.} 

The behavioral restriction of this model is that given $\yr$, the firm chooses $\yu$ to maximize restricted profits, taking prices $\pu$ as given. In the special case where $\yr$ is not present, this is the usual profit maximization setup. When $\yr$ consists of inputs, this covers revenue maximization. When $\yr$ consists of outputs, this is cost minimization once we interpret negative $\yu$ as inputs and write
\[
\max_{\yu :\: (\yr, \yu) \in \Y(\e)} \p_z \tr \yu = - \min_{\yu:\: (\yr, \yu) \in \Y(\e)} \p_z \tr (-\yu).
\]
We emphasize that throughout, $\yr$ can be a vector, and so we cover cost minimization with multiple inputs, and revenue maximization with multiple outputs.

Overall, we consider firms that are price-taking in the variables of choice $\yu$, and study a static problem without uncertainty. We note though that in principle the production set $\Y(\e)$ is general enough to describe paths of production possibilities throughout time, as would arise if there is investment.

\subsection{Setting and Data}\label{subsec: setting and data}
We study identification in settings in which an analyst observes many realizations of certain values of the restricted profit maximization problem as prices vary. In the most general version, we observe noisy measurements of restricted profits, which are the values of the restricted problem. Specifically, we consider the setup
\[
\rands{\pi}_r = \pi_r (\rand{\y}_{-z}, \rand{\p}_{z}, \rand{\e})+\rands{\eta}\:\as,
\]
where $\rands{\pi}_r$ and $\rand{\y}_{-z}$ are observed,\footnote{We use bold font for random variables and random vectors and regular font for their realizations.} $\rands{\eta}$ is unobserved measurement error, and $\rand{\e}$ is unobservable productivity level. For each component of $\rand{\p}_{z}$, the analyst either observes the corresponding price, or more generally observes a price proxy $\rand{x}_j$ that is linked to the unobserved price by the relationship $\rand{p}_{z,j} = g_j(\rand{x}_j, \tilde{\rand{x}})$, where $\tilde{\rand{x}}$ consists of some control variables. We provide further examples and discussion of such proxies in Section~\ref{sec:homog}.
\par
As an example of observables for cost minimization of hospitals \citep{bilodeau2000hospital}, the analyst observes total cost on variable inputs $\rand{\y}_{z}$ (labor, supplies, food for patients, drugs, and energy), input prices or input-price proxies, fixed outputs (inpatient care and outpatient visits), and the fixed inputs (number of physicians and capital). We emphasize that we do not need to observe the quantities $\rand{\y}_{z}$ of the flexibly chosen variables.\footnote{As discussed in the introduction, for additional data sets, see \citet{nerlove63,roberts1996output,costelectricityKira2007,foster2008reallocation,epple2010new, doraszelski2013r,albouy2018housing,burke2019sell}, and \citet{combes2017production}.}

Now we turn to the description of the sources of variation in our setup. Although we do not fully flesh out an equilibrium model incorporating selection, we provide an informal discussion of these forces. First, prices can vary across markets due to variation in endowments or the income or tastes of consumers. Our results apply when an analyst observes a single firm from each market, and has observations from many markets. Our results also apply when an analyst observes multiple firms in each market. We focus on the former case to simplify presentation, so that we can avoid market-level subscripts.

\section{Recoverability of Restricted Profit Function}\label{sec: restricted prof func}
Our ultimate goal is to learn about the production correspondence. We proceed in three steps. In this section, we first identify the restricted profit function (or value function) for heterogeneous firms assuming that the prices are perfectly observed. In Section~\ref{sec:homog} we show how to apply our analysis to the general case with unobserved prices. 
In subsequent sections we show how to use information on the restricted profit function to recover features of the production correspondence and describe the most that can be learned concerning counterfactual questions.

Identifying the restricted profit function for heterogeneous firms is challenging. The value function is nonseparable in latent productivity. Both the restricted variables $\rand{\y}_{-z}$ and prices $\rand{\p}_{z}$ may be endogenous. This leads to simultaneity and selection biases. We consider a setting without panel data or instruments. We present a new technique to identify the restricted profit function that addresses these challenges. The key restrictions of the technique are that (i) heterogeneity is one dimensional and allows us to rank firms, and (ii) there are finitely many types of firms. 

\subsection{Production Monotonicity}\label{subsec: monotonicity and independence}
It is well-known that the firm problem admits a representative agent, and in principle this observation can be used to recover a \textit{representative agent} restricted profit function. Even a representative agent analysis here is nontrivial because of challenging selection/simultaneity issues discussed previously. Here, we wish to recover not only a representative agent restricted profit function, but also recover the heterogeneous structural restricted profit functions. Recovering heterogeneous structural functions allows us to a conduct rich counterfactual analysis concerning how different types of firms are differentially affected by a policy. 
\par
To get traction on this problem, we assume firms are rankable in terms of productivity. We think of heterogeneous productivity as an ability to produce more with a given level of inputs (or produce the same output using lower levels of inputs). In other words, the production set of a firm with lower value of productivity is a subset of the production set of a firm with a higher productivity (see Figure~\ref{fig:nested sets}). Note that $\Y_r(\yr,\e)\subseteq\Y_r(\yr,\tilde\e)$ if and only if $\pi_r(\yr,\pu,\e)\leq \pi_r(\yr,\pu,\tilde\e)$ for all $\pu$. This means that more productive firms have access to a bigger set of production possibilities, and will make more profits or pay lower costs given prices. We formalize this monotonicity by the following ranking assumption on the restricted profit function.
\begin{figure}
\begin{center}
\begin{tikzpicture}[scale=0.6]
\draw[thick,->] (-2,0) -- (6,0) node[anchor=north west] {$\y_i$};
\draw[thick,->] (5,-1) -- (5,4) node[anchor=south east] {$\y_o$};
\draw[-,thick] (5,0) to [out=110,in=-5] (-1,4) to [out=-5,in=-1] (-2,4.1);
\draw[-,thick] (5,0) to [out=100,in=-5] (-1,5) to [out=-5,in=-1] (-2,5.1);
\draw[thick,->] (-1,4) -- (-1,4.9);
\draw[thick,->] (1,3.6) -- (1,4.4);
\node at (-3,5.1) {$Y(\tilde\e)$};
\node at (-3,4.1) {$Y(\e)$};
\end{tikzpicture}
\end{center}
\caption{Nested Production Sets. $\tilde\e>\e$. \label{fig:nested sets}}
\end{figure}
\begin{assumption}[Strict Monotonicity] \label{assm:strict monot}
For every $\yr$, $\pu$, $\e$, and $\tilde\e$ in the support, if $\e<\tilde\e$, then $\pi_r(\yr,\pu,\e)<\pi_r(\yr,\pu,\tilde\e)$.
\end{assumption}
Strict monotonicity of structural functions has been considered previously in e.g. \citet{matzkin2003nonparametric}. Assumption~\ref{assm:strict monot} is satisfied in many settings. For instance, it is satisfied in a standard single output production function setting with Hicks-neutral productivity. To be more specific, let the single output be $y_o$ and let inputs be $l$ and $k$, interpreted as labor and capital. Then the set $\Y(\e)$ is described by tuples $(y_o, -l, -k)$ that satisfy $y_o \leq f(l, k, \e)$, where $f$ is the production function. If $f(l,k,\e) = A(e) \bar{f}(l,k)$ for some nonnegative, strictly increasing function $A$, and $\bar{f}$ is a nonnegative strictly convex function, then $f(l, k, \e)$ is strictly increasing in $\e$. In this case, $\pi(\p,\cdot)$ satisfies Assumption~\ref{assm:strict monot}.

More generally, the function $f(l, k, \e) = A_o(\e) \bar{f}(A_l(\e) l, A_k(\e) k)$ for strictly increasing functions $A_o$, $A_l$, and $A_k$ fits into our setup.\footnote{\citet{li2017constructive} study a related setup with random coefficients Cobb-Douglas technology, imposing that the ratio of random coefficients is a monotone function of a single latent scalar random variable.} A more general setup would allow a different shock to enter $A_o, A_l$, and $A_k$ (e.g. \citealp{doraszelski2018measuring}) and would be outside of our framework. Overall, while Hicks-neutral heterogeneity is a special case of our framework when there is a \textit{single} output, it is considerably more restrictive than needed for the monotonicity assumption to hold.

The assumption that production sets are nested in $\e$ is equivalent to the profit function being weakly increasing in $\e$. Thus, value functions are the ``right'' structural function in which to impose monotonicity if we think of higher productivity as leading to more production possibilities. One may draw the intuition that in general other structural functions are monotone in unobservable heterogeneity. This intuition is false without more structure.

\begin{example}[Nonmonotonicity of Inputs/Outputs] \label{ex:nonmonotone-supply}
Consider the production sets depicted in Figure~\ref{fig:nonmonotonicsupply}. Each production set is given by $Y(e_i)=\{(y_o,l)\tr\in\Real\times\Real_{+}\::\:y_o\leq f(l, e_i)\}$, where  $f(l, e_1)<f(l, e_2)<f(l, e_3)$ for all $l>0$. Here, $\pi(p,e_1)<\pi(p,e_2)<\pi(p,e_3)$ for all positive $p$ and Assumption~\ref{assm:strict monot} is satisfied. Given the price vector $p=(p_o,p_k)\tr$ in Figure~\ref{fig:nonmonotonicsupply}, the optimal levels of inputs and outputs are nonmonotone in productivity since $l^*(p,e_1)< l^*(p,e_3)< l^*(p,e_2)$ and $y^*_{o}(p,e_1)< y^*_{o}(p,e_3)< y^*_{o}(p,e_2)$. For a numerical example see Online Appendix~C.
\begin{figure}
	\begin{center}
		\begin{tikzpicture}[scale=.7]
		\draw[-,thick] (5,-2) to [out=110,in=-40] (4,-.6) to [out=150,in=-1] (-2,1);
		\draw[-,thick] (5,-2) to [out=100,in=-38] (1.5,2.2) to [out=150,in=-1] (-2,3.2);
		\draw[-,thick] (5,-2) to [out=95,in=-40] (3.5,1.75) to [out=150,in=-1] (-2,3.5);
		\draw[dashed] (6,-2) -- (-1,3);
		\draw[dashed] (6,-.9) -- (-1,4);
		\draw[dashed] (6,0) -- (-1,5);
		\draw [fill=black] (4,-.6) circle[radius=.1];
		\draw [fill=black] (1.5,2.2) circle[radius=.1];
		\draw [fill=black] (3.5,1.75) circle[radius=.1];
		\node at (-3,1) {$e_1$};
		\node at (-3,3.1) {$e_2$};
		\node at (-3,3.6) {$e_3$};
		\node at (7.5,-2) {$\pi(p,e_1)$};
		\node at (7.5,-0.9) {$\pi(p,e_2)$};
		\node at (7.5,0) {$\pi(p,e_3)$};
		\end{tikzpicture}
	\end{center}
	\caption{Nonmonotonic supply. \label{fig:nonmonotonicsupply}}
\end{figure}
\end{example}
Failures of monotonicity in the optimal choice of input or output have been discussed as well in \citet[Section 4]{pakes1996dynamic}. Thus, rather than focus on the structural functions describing optimal input/output choices, this paper focuses instead on the restricted profit function, which \textit{is} monotone in a scalar unobservable under the assumption that production sets are nested in $\e$.

\subsection{Discrete Heterogeneity and Monotone Selection}
With this setup, we consider a new technique to identify the restricted profit function allowing endogeneity.
The reason endogeneity is a central concern in such problems is that constraints may be endogenous. For example, in the cost minimization problem, output ($\rand{\y}_{-z}=\rand{\y}_o$) is typically a choice variable for the firm. An additional endogeneity concern is that firms may choose in which markets to operate. This can induce a selection issue, though we emphasize that once a market is chosen, the input/output vector is determined taking market prices as fixed. As discussed in Section~\ref{subsec: setting and data}, price variation in our setting arises because firms operate in different markets, which have different endowments or consumer tastes.
\par
The key restriction we impose is that there are finitely many types of firms. We formalize this as follows.
\begin{assumption}[Finite Heterogeneity]\label{assm:disc heter}
$E=\{1,2,\dots,d_{e}\}$, where $d_{e}$ is finite and unknown to the researcher.
\end{assumption}
This assumption allows us to identify structural functions without instruments. If instruments are available, continuous heterogeneity can be tackled by existing techniques \textit{provided} there is no measurement error; see for example Online Appendix~B. We emphasize that heterogeneity here is in terms of the production types, but due to measurement error in the data we may see continuous distributions of the restricted values, even when we condition on all other observables. In this modeling decision we are close to structural dynamic discrete choice literature that often assumes unobserved discrete heterogeneity that is smoothed out by some continuous idiosyncratic noise (e.g. extreme value distributed preference shock). See, for instance, \citet{arcidiacono2011conditional}.\footnote{For applications of discrete unobserved heterogeneity, see \citet{fox2016nonparametric} in multinomial choice and \citet{bonhomme2015grouped} with panel data.} We are not aware of any identification results that allow for both measurement error and continuous nonseparable structural unobserved heterogeneity in cross-sectional data.

We allow rich selection into markets, but impose a monotonicity restriction relating the types of firms that can be present, conditional on certain observables.
\begin{assumption}[Monotone Presence] \label{assm: monotone selection}
\[
\Prob{\rand{\e}=\e|\rand{\y}_{-z}=\yr,\rand{\p}_z=\pu}>0\implies\Prob{\rand{\e}=\tilde\e|\rand{\y}_{-z}=\yr,\rand{\p}_z=\pu}>0
\]
for all $\yr$, $\pu$, $\e$, and $\tilde\e$ in the support such that $\e<\tilde\e$.
\end{assumption}
This means that if we see a firm of type $\e$ active in \textit{some} market and producing $y_{-z}$, conditional on $\rand{\p}_z=p_z$, then for any higher productivity $\tilde\e$, there is some market in which the higher type is active at the same value of conditioning variables. In principle, this other ``market'' could be the same market in which the firm with productivity $\e$ is present. The key restriction is that since we also condition on quantities, we need the higher type to \textit{also} produce the same quantities.

As an example, consider  the (unrestricted) profit function. Suppose entry depends on whether a firm obtains nonnegative profits. Specifically,
\[
\e \text{ enters } \iff \pi(\p,\e) \geq 0,
\]
where there are no restricted variables. Since we assume monotonicity of $\pi$ in $\e$, this is a monotone threshold rule, and satisfies Assumption~\ref{assm: monotone selection}.

Assumption~\ref{assm: monotone selection} is considerably more general than a one-sided selection rule. Importantly, it is only about the \textit{support} of $\rand{\e}$ conditional on some other variables. The reason we require this is that while reasonable selection rules into \textit{markets} may result in a one-sided threshold rule, here we also need to allow selection into the quantities of the restricted variables $y_{-z}$. For example, as $\e$ increases the optimal quantity of the restricted variables may change. Assumption~\ref{assm: monotone selection} allows this and is satisfied if, for example, there are other unobserved variables that shift the optimal choice of restricted variables $y_{-z}$ (e.g. unobserved prices of the restricted variables).

\subsection{Identification}
We now turn to identification of the restricted profit function. First, recall that we observe potentially mismeasured restricted profits:
\[
\rands{\pi}_r = \pi_r (\rand{\y}_{-z}, \rand{\p}_{z}, \rand{\e})+\rands{\eta}.
\]
If $\rands{\eta}$ is independent of  $(\rand{\y}_{-z}, \rand{\p}_z, \rand{\e})$, then Assumption~\ref{assm:disc heter} implies that the conditional distribution of $\rands{\pi}_r$ can be written as a finite mixture of shifted distributions of $\rands{\eta}$:
\begin{align*}
    F_{\rands{\pi}_r|\rand{\y}_{-z}, \rand{\p}_z}(\cdot|y_{-z},p_z)=\sum_{e\in E}F_{\rands{\eta}}(\cdot-\pi_r(y_{-z},p_z,e))\Prob{\rand{e}=e|\rand{y}_{-z}=y_{-z},\rand{p}_z=p_z},
\end{align*}
where $F_{\rands{\pi}_r|\rand{\y}_{-z}, \rand{\p}_z}(\cdot|y_{-z},p_z)$ is the conditional cumulative distribution function (c.d.f.) of $\rands{\pi}_r$ conditional on $\rand{y}_{-z}=y_{-z}$ and $\rand{p}_z=p_z$, and $F_{\rands{\eta}}$ is the c.d.f. of $\rands{\eta}$. There are numerous ways to identify the above finite mixture model under different sets of assumptions that may be valid in different environments (see, for instance, \citealp{kitamura2018nonparametric} and references therein). However, most of these results use either repeated measurements (i.e. panels) or use variation in conditioning variables, and require some form of exclusion restrictions (e.g., some conditioning variables affect $\pi_r(y_{-z},p_z,e)$ but do not affect $\Prob{\rand{e}=e|\rand{y}_{-z}=y_{-z},\rand{p}_z=p_z}$), or the presence of instruments. We propose a new set of assumptions to identify the above finite mixture in cross-sections, without instruments and exclusion restrictions. Moreover, our approach is constructive and the assumptions are easy to interpret. 

Let $\Delta\pi_r(\yr,\pu,\e) = \pi_r(\yr,\pu,\e)-\pi_r(\yr,\pu,\e-1)$ denote the restricted profit difference between firms with adjacent productivity. We impose the following assumption on the measurement error $\rands{\eta}$.
\begin{assumption}\label{assm: measurement error}
\begin{enumerate}
    \item $\rands{\eta}$ is independent of $(\rand{\y}_{-z}, \rand{\p}_z, \rand{\e})$, mean zero, has connected support, and satisfies $\Prob{\abs{\rands{\eta}}\leq K/2}=1$ for some $K<\infty$;
    \item (Separatedness) There exists $(\yr^*,\pu^*,\e^*)$ in their support such that
    \[
    K<\begin{cases}
    \Delta\pi_r(\yr^*,\pu^*,\e^*+1),&\text{ if } \e^*=1,\\
    \Delta\pi_r(\yr^*,\pu^*,\e^*),&\text{ if } \e^*=d_{\e},\\
    \min\left\{\Delta\pi_r(\yr^*,\pu^*,\e^*+1),\:\Delta\pi_r(\yr^*,\pu^*,\e^*)\right\},&\text{ otherwise. }
    \end{cases}
    \]
\end{enumerate}
\end{assumption}
We note that multiplicative measurement error can be handled by similar independence and separatedness assumptions.\footnote{The bounded support and separatedness conditions in Assumption~\ref{assm: measurement error} can be relaxed using results in \citet{schennach2016recent} if one has access to repeated cross-sections.}

Assumption~\ref{assm: measurement error}(i) means that the measurement error is classical. It also imposes a location normalization on the boundedly-supported measurement error.\footnote{For examples of papers studying boundedly-supported measurement errors see \citet{hu2010deconvolution,d2010identification}, and \citet{hu2017injectivity}.} The bounded support assumption is empirically relevant in many settings. For instance, revenues and costs cannot be negative, which provides a one-sided bound. Assumption~\ref{assm: measurement error}(ii) is more substantial. This assumption imposes that the gap between the structural profits of the types adjacent to $e^*$ must be sufficiently small compared with the support of measurement error. This can be restrictive in certain empirical settings but is essential for this method. We argue that boundedness and separatedness are appropriate in our empirical illustration in Appendix~\ref{app: application}.

Note that Assumption~\ref{assm: measurement error}(ii) has to be imposed on one triplet $(\yr^*,\pu^*,\e^*)$ only. Thus, in general, the measurement error may completely change the ranking of restricted profits. Moreover, this triplet does not need to be known. A simple sufficient condition for Assumption~\ref{assm: measurement error}(ii) that uses shape restrictions of the restricted profit function is stated in the following result.
\begin{lem}[Rich Support]\label{lemma: rich support}
If Assumption~\ref{assm:strict monot} holds and there exist $\yr^*$ and $\pu^*$ such that $\cup_{\lambda\geq 1}\{\lambda\pu^*\}$ is in the support of $\rand{\p}_{z}$ conditional on $\rand{\y}_{-z}=\yr^*$, then Assumption~\ref{assm: measurement error}(ii) is satisfied.
\end{lem}
This exploits homogeneity in prices, i.e. $\pi_r(\yr^*,\lambda\pu^*,\e)=\lambda\pi_r(\yr^*,\pu^*,\e)$ for all $\e$ and $\lambda>0$. The idea behind Lemma~\ref{lemma: rich support} is that although the difference between profits evaluated at a particular price may not be big enough to offset the effect of the measurement error (e.g. $\Delta\pi_r(\yr^*,\pu^*,\e^*+1)\leq K$), by exploiting homogeneity we always can find $\lambda^*$ big enough such that 
\[
\Delta\pi_r(\yr^*,\lambda^*\pu^*,\e^*+1)=\lambda^*\Delta\pi_r(\yr^*,\pu^*,\e^*+1)>K.
\]
The conditions of Lemma~\ref{lemma: rich support} guarantee that an extreme price $\lambda^*\pu^*$ can be found in the support for every finite $K$. Thus, the support of prices does not have to be unbounded, just sufficiently large relative to the initial difference. 

Now we can state our main identification result for the restricted profit function.
\begin{thm}\label{prop: profit func discret I}
Suppose Assumptions~\ref{assm:strict monot}-\ref{assm: measurement error} hold. Then using $F_{\rands{\pi}_r|\rand{y}_{-z},\rand{p}_{z}}$, $\pi_r$ is identified over the joint support of $\rand{y}_{-z}$, $\rand{\p}_{z}$, and $\rand{\e}$.
\end{thm}
Here, we may not be able to identify the structural restricted profit function for certain arguments outside of the support. This is particularly relevant for low types; there may be many combinations of prices and quantities such that low types do not produce either because it is infeasible for them or unprofitable.

Importantly, Theorem~\ref{prop: profit func discret I} only imposes a mild restriction on the stochastic dependence between unobservable heterogeneity $\rand{\e}$ and observed $\rand{\y}_{-z}$ and $\rand{\p}_z$. In particular, in cost minimization settings, the output level and input prices can be related to the distribution of productivity in flexible ways. What is key is the monotonicity restriction on selection into markets described in Assumption~\ref{assm: monotone selection}.
\par
The intuition behind Theorem~\ref{prop: profit func discret I} is that without restricting the dependence structure, monotonicity in the restricted profit function implies that firms always can be ranked.  The assumption of the discrete heterogeneity allows us to match firms with the same ranking across different markets, and thereby construct the restricted profit function.
\par
Theorem~\ref{prop: profit func discret I} can be used to weaken assumptions usually made in analysis of restricted profit maximizing behavior. For instance, with cost minimization, \citet{bilodeau2000hospital} focuses on a parametric setup with additively separable heterogeneity and assumes that fixed variables are exogenous. While working with the same observables, our methodology does not require parametric restrictions, and does not assume exogeneity.

\begin{rem}[Testability]
Theorem~\ref{prop: profit func discret I} identifies the restricted profit function $\pi_r$ without using the shape restrictions that characterize such functions. Thus, the assumptions in this paper are testable. Specifically, for each $\e$, the identified function $\pi_r(\yr,\p_z,\e)$ must be convex, monotonically decreasing, and homogeneous of degree $1$ in the prices of the flexible variables $\p_z$. These implications can be tested with data on the values of the restricted problem $\rands{\pi}_r$, the restricted quantities $\rand{\y}_{-z}$, and prices $\rand{\p}_z$.
\end{rem}

\section{Unobservable Prices and Proxies}\label{sec:homog}
In Section~\ref{sec: restricted prof func}, we showed how to identify the restricted profit function when the entire vector of prices of flexibly chosen variables, $\rand{\p}_{z}$, is observed. In many empirical applications not all prices are observed. This may cause concern about \textit{omitted price bias} (see \citealp{zellner1966specification,klette1996inconsistency,katayama2003plant}, and \citealp{epple2010new}). However, the researcher may have access to some observable proxies that are informative about unobservable prices. For example, the rental rate of capital may be linked to market-specific characteristics such as short-term and long-term interest rates. Wages may be linked to the unemployment level or aggregate labor supply. \citet{de2016prices} uses output price, market shares, product dummies, firm location, and export status as proxies for unobservable input prices. In the housing market, an analyst may use location as a price proxy for a house as in \citet{combes2017production}.\footnote{Hedonic pricing models also exhibit similar structure. However, in that literature it is assumed that both prices and proxies are observed. See, for instance, \citet{ekeland2004identification}.}
\par
This section studies how to identify the function linking prices proxies to unobserved prices through
\[
\rand{\p}_{z,j} = g_j (\rand{\x}_j,\tilde{\rand{\x}})\:\as,
\]
where $g_j$ is an unknown function and $\rand{\p}_{z,j}$ is a component of the vector of prices $\rand{\p}_{z}$ of the flexibly-chosen variables. We show how to identify $g_j$ using the fact that the restricted profit function is homogeneous of degree $1$, though as discussed in the Introduction, the technique we present is new and applies to any degree of homogeneity.\footnote{Homogeneity has been used for identification in \citet{matzkin1992nonparametric}, which differs in techniques and setting.} We assume that every price has its own \emph{excluded} proxy $\rand{\x}_j$, which is a proxy that affects its own price and does not affect any other prices. The vector of common proxies $\tilde{\rand{\x}}$  may include common market characteristics such as size of the market or other macroeconomic characteristics. Importantly, since $g_j$ is fully nonparametric, $\tilde{\rand{\x}}$ can include categorical variables such as location (e.g. country or state) and time (e.g. month or year) identifiers. The special case in which price is observed corresponds to $g_j(\x_j,\tilde\x)=x_j$, where $x_j$ is the price of $y_j$. To simplify the exposition we drop $\tilde{\rand{\x}}$ from the notation, and analysis may be interpreted conditional on $\tilde{\rand{\x}}=\tilde{\x}$. For instance, we write $g_j(\rand{\x}_j)$ instead of $g_j (\rand{\x}_j,\tilde{\rand{\x}})$. We denote $\x=(\x_j)_{j=1,\dots,d_{\y_{z}}}\in\X$ and $g(\x) = (g_j(\x_j))_{j=1,\dots,d_{\y_{z}}}$.
\par
Note that we assume prices are not a function of $\e$ or any other unobservables. Importantly, this rules out measurement error in prices. In our setup prices vary across markets, but are constant within a given market. Price-taking behavior implies that prices can be a function of the distribution of $\rand{\e}$ in a market, but not the firm-specific productivity $\e$.
\par
We first present an informal outline how to identify $g$ when one observes unrestricted profits, so that there are no restricted variables and the subscript $z$ can be dropped. If the function $g$ were known, we could identify $\pi$ directly by previous arguments. What remains is to identify $g$.
Recall that the profit function $\pi(\cdot,\e)$ is homogeneous of degree $1$, which from Euler's homogeneous function theorem yields the system of equations
\[
\sum_{j = 1}^{d_\y} \partial_{p_j} \pi(\p,\e)\p_j = \pi(\p,\e)\,.\footnote{Recall that we work with the unrestricted profit function for notational simplicity, but the restricted profit function is also homogeneous of degree $1$ in prices.}
\]
Replacing prices with price proxies, we obtain
\begin{equation} \label{eq:eulers1}
    \sum_{j=1}^{d_\y}\partial_{p_j}\pi(g(\x),\e)g_{j}(\x_j)=\pi(g(\x),\e)\,.
\end{equation}
Define $\tilde{\pi} (\x,\e) = \pi(g(\x),\e)$. Because $x_j$ is exclusive to $p_j$, the cross-partial derivatives satisfy $\partial_{x_j}{g_k(x_k)}=0$ for $j\neq k$. We thus have
\[
\partial_{\x_j}\tilde\pi(\x,\e)=\sum_{k}\partial_{p_k}\pi(g(\x),\e)\partial_{x_j}g_{k}(\x_k)=\partial_{p_j}\pi(g(\x),\e)\partial_{x_j}g_{j}(\x_j)\,.
\]
Plugging this in to (\ref{eq:eulers1}) we obtain
\begin{align} \label{eq:eulers2}
    \sum_{j=1}^{d_\y}\partial_{x_j}\tilde{\pi}(\x,\e) \frac{g_{j}(\x_j)}{\partial_{x_j} g_j(x_j)}=\tilde{\pi}(x,\e)\,.
\end{align}
Assume for now that $\tilde{\pi}(\cdot,\e)$ is identified. Thus the only unknowns involve $g$. By varying $x$, holding everything else fixed, Equation~\ref{eq:eulers2} can be used to generate a system of equations. We show that when a certain rank condition is satisfied, it is possible to identify the entire function $g$ using an appropriate scale/location normalization. We note that if all prices are observed except one, then we may directly apply Equation~\ref{eq:eulers2} to learn about $g_j$.

To formalize this, we impose location/scale conditions and some regularity conditions on $g$. 
\begin{assumption}\label{assm:hedonic function}
\begin{enumerate}
    \item $g_{1}(\x_{1})=\x_1$ for all $\x_{1}$, i.e. the price of the $1$-st flexibly chosen variable is observed;
    \item The value of $g$ is known at one point, i.e. there exist known $\x_0$ and $\p_0$ such that $g(x_0)=\p_0$;
    \item $X = \prod_{j = 1}^{d_{\y_z}} X_j$ where each set $X_j\subseteq\Real$ is an interval with nonempty interior;
    \item $g_{j}(\cdot)$ is continuous everywhere and differentiable on the interior of $X_j$, and the set
    \[
    \left\{\x_j\in\X_j\::\:\partial_{\x_j}g(\x_j)=0\right\}
    \] 
    has Lebesgue measure zero for every $j$.
\end{enumerate}
\end{assumption}
Assumptions~\ref{assm:hedonic function}(i)-(ii) allow us to identify the scale and the location, respectively, of the multivariate function $g$. Since we can always relabel both outputs and inputs, Assumption~\ref{assm:hedonic function}(i) is equivalent to assuming that at least one price (not necessary $\p_{1}$) is observed.
\par
We now turn to our rank condition. This condition ensures that the system of equations generated from (\ref{eq:eulers2}) has sufficient variation to recover terms such as $g_j(x_j) / \partial_{x_j} g_j(x_j)$.
\begin{defn}\label{def:rank condition}
We say that $h:\prod_{j = 1}^{d_{\y_z}} X_j\to\Real$ satisfies the rank condition at a point $\x_{-1}\in \prod_{j = 2}^{d_{\y_z}} X_j$ if there exists a collection $\{t_l\}_{l=1}^{d_{\y_z}-1}\subseteq X_1$ such that
\begin{enumerate}
    \item $x^*_l=(t_l, x_{-1}\tr)\tr\in\prod_{j = 1}^{d_{\y_z}} X_j$;
    \item The square matrix
    \begin{equation*}
    \left[\begin{array}{ccc}
        \partial_{x_2}h(\x^*_{1}) &
        \dots&
        \partial_{x_{d_{\y_z}}}h(\x^*_{1}) \\
        \partial_{x_2}h(\x^*_{2}) &
        \dots&
        \partial_{x_{d_{\y_z}}}h(\x^*_{2}) \\
        \dots&\dots&\dots\\
        \partial_{x_2}h(\x^*_{d_{\y_z}-1}) &
        \dots&
        \partial_{x_{d_{\y_z}}}h(\x^*_{d_{\y_z}-1})
    \end{array}
    \right]
\end{equation*}
is nonsingular.
\end{enumerate}
\end{defn}
We will apply this rank condition to $\tilde{\pi}$ in place of $h$. It is helpful to recall that by Hotelling's lemma, partial derivatives of $\tilde{\pi}$ take the form
\[
\partial_{x_j} \tilde{\pi}(x,\e) =  \partial_{p_j} \pi (\p,\e ) |_{\p = g(x) } \partial_{x_j} g_j(x_j) =\y_j (g(x),\e) \partial_{x_j} g_j(x_j)\,,
\]
where $\y_j(g(x),\e)$ is the supply for good $j$. Thus, this rank condition applied to $\tilde{\pi}$ may equivalently be interpreted as a rank condition involving the supply function for the goods as well as certain derivatives of $g$. In words, variation in observed prices should induce enough variation in supply of goods with unobserved prices. 

The following result provides conditions under which the price-proxy function $g$ is identified. We note that while our exposition above covered the case of unrestricted profits, the following result holds for the more general setting of restricted profits. Thus, instead of the function $\tilde\pi$, we will use its restricted version defined via $\tilde\pi_r(\x,\e) = \pi_r(\y^*_{-z},g(\x),\e)$, where $y^*_{-z}$ is fixed.
\begin{thm}\label{thm: attributes and profits}
Suppose Assumption~\ref{assm:hedonic function} holds. Then $g$ is identified over the support of $\rand{x}$ if for some $\y^*_{-z}$, the following conditions hold:
\begin{enumerate}
\item  $\tilde\pi_r(\x,\e)$ is identified for each $\x$ and $\e$ in the support; 
\item For every $\x_{-1}\in\prod_{j = 2}^{d_{\y_z}} X_j$, there exists $\e^{**}$ in the support such that $\tilde\pi_r(\cdot,\e^{**})$ satisfies the rank condition at $x_{-1}$.
\end{enumerate}
\end{thm}
To interpret (i), recall that Theorem~\ref{prop: profit func discret I} provides conditions under which $\tilde{\pi}_r$ is identified from the conditional distribution of $\pi_r(\rand{\y}_{-z},g(\rand{x}),\rand{e})$ conditional $\rand{\x}=\x$ and $\rand{\y}_{-z}=\y_{-z}$. To apply those results one just needs to replace $\rand{\p}_z$ by $\rand{\x}$. Here we highlight that given \textit{some} way to identify a structural function of the form of $\tilde{\pi}_r$, we can identify $g$. Thus, if a researcher has another means of identifying the structural function $\tilde{\pi}_r$, then this theorem can be applied.
\par
Part (ii) requires sufficiently rich variation in the reduced form profit function $\tilde{\pi}_r$ for \textit{some} value of productivity $e^{**}$. To further interpret the rank condition, we study it in two parametric examples in Online Appendix~D. There we show that the rank condition can be satisfied for the \citet{diewert73genleontief} profit function, but can fail for every possible parameter value with Cobb-Douglas technology. The reason Cobb-Douglas fails is that its profit function is additively separable when logs are taken.
\begin{rem}[Other Degrees of Homogeneity]
It is straightforward to generalize our technique to a homogeneous function of any degree $\alpha\geq0$ since the main identifying equation (\ref{eq:eulers2}) can be rewritten as
\begin{equation} \label{eq:homog alpha}
  \sum_{j=1}^{d_\y}\partial_{x_j}\tilde{\pi}(\x,\e) \frac{g_{j}(\x_j)}{\partial_{x_j} g_j(x_j)}=\alpha\tilde{\pi}(x,\e)\,.
\end{equation}
Here we study the restricted profit function, so $\alpha = 1$, but an analogous equation holds for other homogeneous structural functions. As one example, recall the supply function is homogeneous of degree $0$ in prices for a price-taking, profit-maximizing firm.
\end{rem}

\begin{rem}[Aggregation]
The key shape restriction used for identification in this section is homogeneity of a structural function. Importantly, homogeneity is a shape restriction that is preserved under expectations. Note that while we use homogeneity of degree $1$ here, this is true for any degree of homogeneity. See in particular Equation~\ref{eq:homog alpha}, which has structure that is preserved under expectations. For this reason, our results work as well with a representative agent analysis involving mean structural demand. We formalize this in Online Appendix~G.
\end{rem}

\subsection{Value as Proxy} \label{subsec:housing}

This section shows how to interpret \citet{epple2010new} through the lens of price proxies. Specifically, we show that average house values in a market can be used as a proxy for a missing output price. We use this setup as well in the empirical illustration in Appendix~\ref{app: application}.
\par
\citet{epple2010new} consider the production of housing in which all goods and services provided by a house are treated as a single output. The analyst observes total revenue of selling a house, and the price of land. Variation in these observables is driven by market variation. Importantly, output and its price are \textit{both} unobserved. Each source of unobservability is recognized as an important problem for the measurement of housing production. Building on \citet{epple2010new} we show how average values in a market serve as a price proxy for this missing price.
\par 
In contrast to \citet{epple2010new}, who work with a representative firm, we study identification in the presence of heterogeneity. As in \citet{epple2010new} we assume constant returns to scale in land and materials, so we can write
\[
y_{o}=f(m,e),
\]
where $f$ is the production function per-acre, and output $y_o$ and materials $m$ are in units per acre (land). The production set associated with this production function is $\Y(e) =\{(\y_{o},-m):\y_{o}\leq f(m,e)\}$. Firms treat land as pre-determined and choose $m$ and $y_o$. We work with the profit function per-acre, written as
\[
\pi(\p_{o},\p_m,\p_{l},\e)=\max_{(\y_{o},-m)\in \Y(\e)}\p_{o}\y_{o}-\p_m m-\p_{l},
\]
where $p_o$, $p_m$, and $p_l$ are prices of output, materials, and land, respectively.
Since the price of materials is unobserved, \citet{epple2010new} assume that it is the same across markets and equals 1. We will make the same assumption and drop $\p_m$ from the notation. 
\par
Since land is pre-determined, its price $\rand{\p}_l$ does not affect the optimal choice of output or materials. Thus, the value of housing $v(\rand{\p}_{o},\rand{\e})=\rand{\p}_{o}\y_{o}(\rand{\p}_{o},\rand{\e})$ and the average value of housing in a market with price $\rand{\p}_o=p_o$, denoted $\overline{v}(\p_{o})=\int v(p_{o},\e)dF_{\rand{\e}}(\e)$, do not depend on price of land $\rand{\p}_l$. Since $y_o(\p_o,\e)$ is monotone in $\p_o$, the average value $\overline{v}(\p_{o})$ is also monotone in $\p_o$. Importantly, $\overline{\rand{v}}$ is identified when we observe total revenue $\rand{\p}_o \rand{\y}_o$.
\begin{lem}
Suppose the distribution of firm productivity $F_{\rand{\e}}$ is the same across markets and the other assumptions of this section hold. If $\overline{v}(\p_{o})$ is strictly increasing in $\p_o$, then average value of housing per market $\overline{\rand{v}}$ is a price proxy, i.e. there exists a function $g$ such that 
\[
\rand{\p}_o=g(\overline{\rand{v}})\:\as
\]
\end{lem}
This equation is analogous to Equation $6$ in \citet{epple2010new} if we interpret their results as a representative agent analysis.
\par
We note here that by using value as a price proxy for output, if profits were observed and the price of materials ($\rand{p}_m$) varied, we could directly use the average value $\overline{\rand{v}}$ and identify $g$ using Theorem~\ref{thm: attributes and profits}. Here, we do not observe profits and the price of materials is assumed fixed at $1$. We thus impose an addition zero-profit assumption as in \citet{epple2010new}. While that paper assumes a single type of firm, which attains zero profits, we assume that profits are zero \textit{on average} in a given market\footnote{\citet{melitz2014heterogeneous} show that free-entry and constant returns of scale imply that ex-ante expected profits are zero, net of entry cost. Here we can assume entry cost is zero. In equilibrium, firms will have zero-profits on average \textit{just before} firms with negative profits leave the market.}: 
\[
\int\pi(\p_{o},\p_{l},\e)dF_{\rand{\e}}(\e)=\p_o\overline{y}_o(\p_{o})-\overline{m}(\p_{o})-\p_l=0,
\]
where $\overline{y}_o$ and $\overline{m}$ are the realizations of the aggregate output per-acre and the aggregate demand for materials per-acre in a given market. Since $\rand{\p}_l$ and $\overline{\rand{v}}$ are observed, the equilibrium assumption nonparametrically recovers a revenue function from production minus materials cost (recall that $\rand{\p}_m=1\:\as$),
\[
p_l = \tilde\pi(\overline{v}) := g(\overline{v})\overline{y}_o(g(\overline{v}))-\overline{m}(g(\overline{v})).
\]
Moreover, since $g(\overline{v})\overline{y}_o(g(\overline{v}))=\overline{v}$ by definition, we also identify material costs
\[
 \tilde{r}(\overline{v})=-\overline{m}(g(\overline{v})).
\]
We identify the function $g$ since we identify $\tilde\pi(\overline{v})$ and $\tilde\pi(\overline{v})-\tilde{r}(\overline{v})$. In particular, $g$ will solve the following differential equation:
\begin{equation}\label{eq: epple ode}
\dfrac{\partial_{\overline{v}}g(\overline{v})}{g(\overline{v})}=\dfrac{\partial_{\overline{v}}\tilde\pi(\overline{v})}{\tilde\pi(\overline{v})-\tilde{r}(\overline{v})}=\dfrac{\partial_{\overline{v}}\tilde\pi(\overline{v})}{\overline{v}}.
\end{equation}
Knowing $g$ we can identify $y_o(p_o,e)$ for different levels of heterogeneity since the observed $\rand{v}$ is equal to $g(\overline{\rand{v}})y_o(g(\overline{\rand{v}}),\rand{e})$. Thus, our approach generalizes \citet{epple2010new} to allow for unobserved heterogeneity in productivity. For a formal generalization of the results in Section~\ref{sec:homog} to settings with other observables see Online Appendix~F.

\section{Identification of the Production Correspondence}\label{sec: identification general}
In Section~\ref{sec: restricted prof func}, we showed how to identify the restricted profit function allowing endogenous entry and correlation between fixed quantities and productivity, without requiring instruments. Section~\ref{sec:homog} extends this result to settings when some prices are not observed but the analyst has price proxies, and provides examples of such proxies. 
\par
We now focus on how any of these identification results for the restricted profit function can be used to identify the primitive object of interest: the production correspondence. For the sake of notational simplicity from now on, we focus on the profit function though the results can be adapted to the restricted profit function by conditioning on $y_{-z}$.
\par
Recall that we start with identification of the profit function $\pi(\p,\cdot)$ only over the support of prices. For notational simplicity, we work with prices and not price proxies.\footnote{More generally we can identify the profit function over the support of $g(\rand{x})$, where $\rand{x}$ is the vector of price proxies.} The support of prices may consist of all nonnegative numbers, or may be much smaller, i.e. finite. We present a sharp identification result for the production correspondence that covers both cases.
\par
First, we note that $\pi(\cdot, \e)$ is homogeneous of degree $1$ in prices. It is also convex in prices, hence continuous. These features lead to consideration of the following richness assumption, which ensures $\Y(\cdot)$ may be recovered uniquely. Let $\P(\e)$ denote the conditional support of $\rand{\p}$ conditional on $\rand{\e}=\e$ (if $\rand{\p}$ and $\rand{e}$ are independent, then $\P(\e)$ does not vary with $\e$).
\begin{assumption} \label{assm:rich price}
\[
\mathrm{int}\left(\mathrm{cl}\left(\bigcup_{\lambda > 0 }\left\{\lambda \p\::\:\p\in\P(\e)\right\}\right)\right)= \Real^{d_\y}_{++}
\]
for all $e$, where $\mathrm{cl}(A)$ and $\mathrm{int}(A)$ are the closure and the interior of $A$, respectively.
\end{assumption}

\begin{figure}
\begin{center}
\begin{tikzpicture}[scale=0.6]
\draw[thick,->] (0,0) -- (6,0) node[anchor=north west] {$\p_1$};
\draw[thick,->] (0,0) -- (0,6) node[anchor=south east] {$\p_2$};
\draw[-,very thick] (0,0) to [out=90,in=145] (5,4.5);
\draw[-,very thick] (0,0) to [out=0,in=250] (3.5,3.8);
\draw[-,very thick] (3.8,4) to [out=0,in=195] (5,4.5);
\draw [fill=white] (0,0) circle[radius=.1];
\end{tikzpicture}
\end{center}
\caption{The set $\P(\e)$ (depicted by black curve) satisfies Assumption~\ref{assm:rich price} and has an empty interior. Dots represent ``holes'' in the support. Thus, $\P(\e)$ is not a connected set.}\label{fig: rich price degenerate support}
\end{figure}
The set 
\[
\bigcup_{\lambda > 0 }\left\{\lambda \p\::\:\p\in\P(\e)\right\}
\]
consists of all prices where $\pi(\cdot,\e)$ is known because of homogeneity. If that set has ``holes,'' then we can fill them by taking the closure of the set since $\pi(\cdot,\e)$ is convex, hence continuous.\footnote{Beyond continuity, the manner in which convexity affects the data requirements that ensure point identification is subtle, and depends on the shape of $\Y(\cdot)$. We provide an illustrative example in Online Appendix~E.} Assumption~\ref{assm:rich price} means that after we consider the implications of homogeneity and continuity, it is as if we have full variation in prices. Figure~\ref{fig: rich price degenerate support} is an example of a set satisfying this assumption. Another example is the Cartesian product of all natural numbers, $P(\e) = \{1, 2, \ldots \}^{d_\y}$. Thus, Assumption~\ref{assm:rich price} does not impose that the support of $\rand{\p}$ contains an open ball.

\begin{figure}
\begin{center}
\begin{tikzpicture}[scale=0.6]
\draw[thick,->] (-2,0) -- (6,0) node[anchor=north west] {$\y_2$};
\draw[thick,->] (5,-1) -- (5,4) node[anchor=south east] {$\y_1$};
\draw[-,thick] (5,0) to [out=110,in=-5] (-1,4) to [out=-5,in=-1] (-2,4.1);
\draw [dashed] (-2,4.25) -- (2.5,3.5) -- (6,-0.3);
\draw [fill=black] (0,3.9) circle[radius=.1];
\draw [fill=black] (3.9,1.9) circle[radius=.1];
\end{tikzpicture}
\end{center}
\caption{$\tilde\Y(\e)$ and $\Y'(\e)$ for $d_{\y}=2$ and $\P(\e)=\{\p^*,\p^{**}\}$. $\tilde\Y(\e)$ is the area under the dashed lines. $\Y'(\e)$ is the area under the solid curve. Dashed lines correspond to two hyperplanes $p^{* \prime} y=\pi(\p^*,\e)$ and $\p^{** \prime} y=\pi(\p^{**},\e)$. They are tangential to the solid curve.}\label{fig: envelope}
\end{figure}
\par
\begin{thm}\label{thm:identification of Y} 
Let $\pi(\p,\e)$ be identified by some previous argument over the set $\p \in P(\e)$ for all $e$. Moreover, let $\tilde\Y(\cdot)$ be defined via
\[
\tilde\Y(\e)=\left\{\y\in\Real^{d_{\y}}\::\:\p\tr \y \leq \pi(\p,\e),\: \forall \p \in \P(\e) \right\}
\]
for all $\e\in\E$. Then
\begin{enumerate}
\item $\tilde\Y(\cdot)$ can generate the data and for each $\e\in\E$, $\tilde{\Y}(\e)$ is a closed, convex set that satisfies free disposal.\footnote{By generate the data we mean that the profit function induced by $\tilde{Y}$ agrees with the identified profit function $\pi(\p,\e)$ for all $\e \in \E$ and  $\p \in \P(\e)$.}
\item A production correspondence $\Y'(\cdot)$ can generate the data if and only if 
\[
\max_{\y\in\Y'(\e)}\p\tr\y=\max_{\y\in\tilde\Y(\e)}\p\tr\y
\] 
for every $\e\in\E$ and $\p\in\P(\e)$. It follows that for any such $\Y'(\cdot)$, $\Y'(\e) \subseteq \tilde{\Y}(\e)$, for each $\e\in\E$.
\item If Assumption~\ref{assm:rich price} holds, then $\tilde{\Y}(\cdot)$ is the only production correspondence that can generate the data.
\end{enumerate}
\end{thm}

Parts (i) and (ii) of Theorem~\ref{thm:identification of Y} are a sharp identification result stating the most that can be said about the production correspondence under our assumptions. These results are related to \citet{varian1984nonparametric}, Theorem 15.\footnote{The set $\tilde{\Y}(\e)$ is related to the ``outer'' set considered in \citet{varian1984nonparametric}, Section 7. The set $\tilde{\Y}(\e)$ is constructed from price and profit information, however, rather than price and quantity information as in \citet{varian1984nonparametric}.} However, \citet{varian1984nonparametric} works only with finite datasets, which are comparable to having a finite support of prices in our setting. In addition, \citet{varian1984nonparametric} observes prices and quantities while we observe prices and profits. Recall that observing prices and quantities implies observation of profits. Finally, \citet{varian1984nonparametric} does not consider unobservable heterogeneity.

Theorem~\ref{thm:identification of Y}(ii) establishes that $\tilde{\Y}(\cdot)$ is the envelope of all production correspondences that can generate the data (see Figure~\ref{fig: envelope}). We note, however, that $\tilde{\Y}(\cdot)$ may not be a production correspondence because it need not satisfy the recession cone property (recall Definition~\ref{defn:production correspondence}(iii)).\footnote{To see this, suppose that a firm of type $\e \in \E$ has 2-dimensional output/input set, prices are a constant vector $\P(\e) = \{ (1,1)\tr \}$, and profits at that price are given by $\pi( (1,1)\tr, \e) = 0$. Then the set $\tilde{\Y}(\e)$ is $\left\{ \y \in \Real^2 :\y_1 +\y_2 \leq 0 \right\}$. This set induces infinite profits for a price-taking firm whenever $\p_1 \neq \p_2$. Hence, this set violates the recession cone property, which is necessary for the firm problem to have a maximizer since $\tilde{\Y}(\e)$ is closed and nonempty, e.g. \citet{kreps2012}, Proposition 9.7. Note from part (iii), when Assumption~\ref{assm:rich price} holds it follows that $\tilde{\Y}$ is a production correspondence, and thus satisfies the recession cone property.}
\par
Theorem~\ref{thm:identification of Y}(iii) is related to classic work on the identification of a deterministic production set from a deterministic profit function.\footnote{See e.g. \citet{kreps2012}, Corollary 9.18 for a textbook result.} In this paper, however, we begin with the distribution of profits and prices. Part (iii) shows that with this distribution, it is possible to identify the distribution of features of $\Y(\cdot)$, such as the distribution of possible profit-maximizing quantities. We emphasize that this is true even if quantities are unobservable. An additional manner in which (iii) differs from textbook analysis is that, in econometric settings, it is not always natural to assume that all prices are observed ($\P(\e) = \Real^{d_\y}_{++}$). Theorem~\ref{thm:identification of Y} clarifies the variation in prices sufficient for nonparametric identification of production sets. We note that while Assumption~\ref{assm:rich price} is sufficient for point identification of $\Y$, it is not necessary as illustrated in Online Appendix~E.
\par

\begin{rem} \label{rem:inputs}
Our identification analysis does not impose any \textit{a priori} restrictions that certain dimensions of $\Y(\e)$ correspond to inputs, i.e. weakly negative numbers. This additional restriction can be imposed by modifying the set constructed in Theorem~\ref{thm:identification of Y}. Specifically, the set $\tilde{\Y}(\e)$ constructed in this theorem may be intersected with an appropriate half-space that encodes that certain dimensions (corresponding to inputs) must be nonpositive. We note that an analogous restriction for outputs is not informative because of the assumption of free disposal.
\end{rem}


\section{Sharp Counterfactual Bounds} \label{sec:sharpbounds}
Theorem~\ref{thm:identification of Y} makes use of a shape restriction to characterize the identified set of the production correspondence for profit-maximizing, price-taking firms. This shape restriction may be used for a dual purpose of providing sharp counterfactual bounds. This follows a long tradition in revealed preference. \citet{varian1982nonparametric,varian1984nonparametric} has exploited the close connections between empirical content, recoverability of structural functions, and counterfactuals. Recent work in demand analysis building on these connections includes \citet{blundell2003nonparametric, blundell2017individual,allen2019identification}, and \citet{aguiar2018stochastic}. In this section we describe a method to bound objects of interest outside of the support of the data.
\par
Since homogeneity and convexity of the heterogeneous profit function allow us to identify it over $\mathrm{cl}\left(\bigcup_{\lambda > 0 }\left\{\lambda \p\::\:\p\in\P(\e)\right\}\right)$, we can associate the conditional support $\P(\e)$ (of prices condition on $\rand{\e}=\e$) with the set where $\pi(\cdot,\e)$ is identified. That is why, for notational simplicity and in this section only, we assume that $\P(\e)$ is a closed subset of the unit sphere $\mathbb{S}^{d_\y - 1}$ for all $e$, and we consider counterfactual prices with norm normalized to 1. 
\par
We first present a result characterizing quantities consistent with profit maximization. Theorem~\ref{thm:identification of Y}(ii) is the basis for the following proposition.
\begin{prop}\label{prop: identified y}
Let $\P(\e)$ be a finite subset of the unit sphere $\mathbb{S}^{d_\y - 1}$. Given $\P(\e)$ and $\{\pi(\p,\cdot)\}_{\p\in\P(\cdot)}$, the set of output/input functions $\{\y_p(\cdot)\}_{\p\in\P(\cdot)}$ can generate $\{\pi(\p,\cdot)\}_{\p\in\P(\cdot)}$ if and only if
\begin{align*}
&\p\tr \y_\p(\e) = \pi(\p,\e)\,,\quad \forall \p \in \P(\e),\e\in\E\,, \\
&\p^{* \prime}\y_{\p^{*}}(\e) \geq \p^{* \prime} \y_{\p}(\e)\,,\:\:\quad \forall\p,\p^{*}\in\P(\e),\e\in\E\,.
\end{align*}
\end{prop}
The vector $\y_\p(\e)$ is interpreted as a candidate supply vector given price $\p$ and productivity $\e$; it need not be unique and thus may not be equivalent to the supply function. Recall that as discussed in Remark~\ref{rem:inputs}, we do not impose \textit{a priori} restrictions that certain components of $\Y(\e)$ are inputs; this would correspond to imposing additional sign restrictions on the functions $y_p(\cdot)$ described in the proposition.
\par
Proposition~\ref{prop: identified y} essentially states that for each $\e$ there must exist output/input vectors such that the weak axiom of profit maximization holds \citep{varian1984nonparametric}. We note, however, that the primitive observables of our paper are the \textit{distribution} of profits and prices. 
\par
We can adapt Proposition~\ref{prop: identified y} to answer counterfactual questions by considering a hypothetical tuple $(p^c,y_{p^c})$ of prices and quantities. If Proposition~\ref{prop: identified y} applies with these additional counterfactual values, then they are feasible given the theory. In more detail, we present bounds on counterfactual objects, potentially with additional restrictions. The counterfactual values involve a function $C$ of interest. The restrictions involve a function $s$ that depends on the counterfactual price $p^c$ and quantity $y_{p^c}$. We encode the restrictions by the combinations such that $s(p^c, y_{p^c}) = 0$. For instance, if the counterfactual price is fixed to a given vector $\overline{\p}^{\text{c}}$ and no restrictions are imposed on $y_{p^c}$, then $s(\p^{\text{c}},\y_{\p^{\text{c}}})=\p^{\text{c}}-\overline{\p}^{c}$. The upper bound with heterogeneity level $\e$ is given by
\begin{align*}
    \overline{C}(\e)=&\sup_{\p^{\text{c}},\y_{\p^{\text{c}}},\{\y_{\p}\}_{\p\in\P(\e)}}C(\p^{\text{c}},\y_{\p^{\text{c}}})\,,\\
    \text{s.t. }& s(\p^{\text{c}},\y_{\p^{\text{c}}})=0\,,\\ 
        &\p\tr \y_\p = \pi(\p,\e)\,,\quad \forall \p \in \P(\e)\,, \\
        &\p^{* \prime}\y_{\p^*} \geq \p^{* \prime} \y_{\p}\,,\:\:\quad\forall\p,\p^{*}\in\P(\e)\cup\{\p^{\text{c}}\}\,.
\end{align*}
The lower bound is given by
\begin{align*}
    \underline{C}(\e)=&\inf_{\p^{\text{c}},\y_{\p^{\text{c}}},\{\y_{\p}\}_{\p\in\P(\e)}}C(\p^{\text{c}},\y_{\p^{\text{c}}})\,,\\
    \text{s.t. }& s(\p^{\text{c}},\y_{\p^{\text{c}}})=0\,,\\ 
        &\p\tr \y_\p = \pi(\p,\e)\,,\quad \forall \p \in \P(\e)\,, \\
        &\p^{* \prime}\y_{\p^*} \geq \p^{* \prime} \y_{\p}\,,\:\:\quad\forall\p,\p^{*}\in\P(\e)\cup\{\p^{\text{c}}\}\,.
\end{align*}
We provide some examples covered by this general setup. Note that these bounds hold for each $\e$, and thus one may also bound the distribution of $\overline{C}(\rand\e)$ and $\underline{C}(\rand\e)$. We reiterate that these upper and lower bounds apply to prices on the unit sphere, though they may be adapted for prices off the unit sphere as illustrated in the following examples.

\begin{example}[Profit bounds for a counterfactual price] \label{ex:counter1}
Suppose that we are interested in upper and lower bounds for profits at a given counterfactual price $\overline{\p}^{\text{c}}$. When prices $\p^{\text{c}}$ are on the unit sphere, we may specify $C(\p^{\text{c}},\y_{\p^{\text{c}}})=\p^{\text{c} \prime}\y_{\p^{\text{c}}}$ and $s(\p^{\text{c}},\y_{\p^{\text{c}}})=\p^{\text{c}}-\overline{\p}^{\text{c}}$. Then the problem can be simplified to get
\begin{align*}
    &\overline{C}(\e)=\sup_{\y\in\tilde\Y(\e)}\overline{\p}^{\text{c}\prime}\y\,,\\
    &\underline{C}(\e)=\max_{\p\in\P(\e)}\inf_{\y\in\tilde\Y(\e)\::\:\p\tr\y=\pi(\p,\e)}\overline{\p}^{\text{c}\prime}\y\,,
\end{align*}
where $\tilde\Y(\e)$ is the envelope of all production possibility sets consistent with the data defined in Theorem~\ref{thm:identification of Y}.
The above bounds are sharp in the following sense: if $\overline{C}(\e)$ is finite, then it is feasible, i.e. there exists a production set that can generate $\overline{C}(\e)$. If $\overline{C}(\e)$ is not finite, then for any finite level $K$ there exists a production set that can generate $C(\p^{\text{c}},\y_{\p^{\text{c}}}) > K$. Analogous statements hold for the lower bounds $\underline{C}(\e)$.
Recall that we assume the support of prices $\P(\e)$ is a subset of the unit sphere. This may be imposed in empirical settings by replacing prices with normalized prices $\rand{\p}/\norm{\rand{\p}}$. For counterfactual questions involving a price off the unit sphere $\overline{\p}^{\text{c}}$, one can bound counterfactual profits at price $\overline{\p}^{\text{c}}/\norm{\overline{\p}^{\text{c}}}$ and then multiply the upper and lower bounds by $\norm{\overline{\p}^{\text{c}}}$.
\end{example}

\begin{example}[Quantity bounds for a counterfactual price] \label{ex:counter2} Suppose that we are interested in the upper and lower bounds for $u\tr\y_{\p^{\text{c}}}$ for a given counterfactual price $\overline{\p}^{\text{c}}$, where $u$ is a vector. For example, with $u=(1,0,\dots,0)\tr$ we are interested in bounds on the first component of $\y$. Then $C(\p^{\text{c}},\y_{\p^{\text{c}}})=u\tr\y_{\p^{\text{c}}}$ and $s(\p^{\text{c}},\y_{\p^{\text{c}}})=\p^{\text{c}}-\overline{\p}^{c}$.
\end{example}

\begin{example}[Profit bounds for a counterfactual quantity] \label{ex:counter3}
Suppose a regulator is considering imposing a new regulation that the first component of the output/input vector is fixed at $\overline{\y}^{\text{c}}_1$. For example, in analysis of health care \citep{bilodeau2000hospital} a hospital may be required to treat a certain number of patients. To bound profits we may write the objective function as $C(\p^{\text{c}},\y_{\p^{\text{c}}})=\p^{\text{c} \prime} \y_{\p^{\text{c}}}$. The constraint is given by $s(\p^{\text{c}},\y_{\p^{\text{c}}})=\y_{1,p^c}-\overline{\y}_1^{\text{c}}$.\footnote{Note that the problem may not have a solution since the set of parameters that satisfy restrictions may be empty.} Bounds on profits with this quantity may be useful for a regulator wondering whether a hospital of type $\e$ would be profitable with the hypothetical regulation. If the upper bound on profits is negative, the answer is definitively no. If the lower bound on profits is positive, the answer is definitively yes.\footnote{This maintains the assumptions of price-taking, profit-maximizing behavior with a technology that is described by a production correspondence.} An additional question a regulator might ask is which types of firms could still be profitable. This can be addressed by studying functions $\overline{C}(\cdot)$ and $\underline{C}(\cdot)$ as $\e$ varies. Note that the constraints $s$ are general, and inequality constraints may be incorporated as well by using indicator functions.
\end{example}

When $\P(\e)$ is finite, computing bounds in Examples~\ref{ex:counter1} and~\ref{ex:counter2} is straightforward since they are the values of linear programs. Example~\ref{ex:counter3} is also a linear program if we add the additional constraint that the counterfactual price is fixed, $p^c = \overline{p}^c$. In general, the computational difficulty of the bounds $\overline{C}$ and $\underline{C}$ depends on the nature of the objective function and the constraint.

\section{Estimation of Production Sets and Consistency}\label{sec:estimation}

The previous identification results describe how to identify the profit or restricted profit function. Appendix~\ref{app: estimation and simulation} describes one estimator of the restricted profit function, but there are many depending on assumptions concerning exogeneity or whether productivity is discrete or continuous. This section links \textit{any} estimator of the restricted profit function to an induced estimator of the corresponding production set. As in previous section, for notational convenience we work with the profit function, though the analysis applies to the restricted profit function by conditioning. In the restricted case, we would instead estimate the restricted production correspondence.

We now describe how an estimator $\hat{\pi}(\cdot, \e)$ of the profit function may be used to construct an estimator $\hat{\Y}(\e)$ of the production possibility set for a firm with productivity level $\e$. The main result in this section relates the estimation error of $\hat{\pi}$ (for $\pi$) and that of the constructed set $\hat{\Y}$ (for $\Y$). Consistency and rates of convergence results for $\hat{\pi}$ thus have analogous statements for $\hat{\Y}$. 

As setup, we now formalize our notions of distance both for functions and sets. We present our result for a fixed $\e\in\E$. We assume that $\pi(\cdot,\e)$ is identified over $\P(\e)=\P=\Real^{d_\y}_{++}$ (we assume Assumption~\ref{assm:rich price}). Given a fixed $\e\in\E$ and $\hat\pi(\cdot,\e)$, a natural estimator for $\Y(\e)$ is
\[
\hat\Y(\e)=\left\{\y\in\Real^{d_\y}\::\:\p\tr\y\leq\hat\pi(\p,\e),\forall\p\in\P\right\}\,.
\]
This set is a plug-in estimator motivated by Theorem~\ref{thm:identification of Y}.
A commonly used notion of distance between convex sets is the Hausdorff distance. The Hausdorff distance between two convex sets $A, B \subseteq \Real^{d_\y}$ is given by
\[
d_H (A,B) = \max \left\{ \sup_{a \in A } \inf_{b \in B} \norm{a - b}, \sup_{b \in B } \inf_{a \in A} \norm{ a - b} \right\}\,.
\]
Unfortunately, the Hausdorff distance between $\Y(\e)$ and $\hat\Y(\e)$ can be infinite. For this reason we will consider the Hausdorff distance between certain extensions of these sets. The following example illustrates why the original distance may be infinite.
\begin{example}\label{ex: infinite d_H}
Suppose that $d_{\y}=2$ and for some $\e\in\E$,
\begin{align*}
    \Y(\e)&=\left\{\y\in\Real\times\Real_{-}\::\:y_1\leq \sqrt{-y_2} \right\}\,,\\
    \hat\Y^m(\e)&=\left\{\y\in\Real\times\Real_{-}\::\:y_1\leq (1-1/m)\sqrt{-y_2} \right\}\,,\quad m\in\Natural.
\end{align*}
Note that although $\lim_{m\to\infty}(1-1/m)\sqrt{-y_2}=\sqrt{-y_2}$ for every finite $\y_2\leq 0$, the Hausdorff distance between these sets is infinite for every finite $m\in\Natural$. 
\end{example}
Example~\ref{ex: infinite d_H} illustrates a technical concern with the Hausdorff distance that arises because of the unboundedness of production possibility sets. However, in empirical applications one may be interested in production possibility sets in regions that correspond to prices that are bounded away from zero. Thus, instead of working with all possible prices we will work only with certain empirically relevant compact convex subsets of $\Real^{d_{\y}}_{++}$. We consider the Hausdorff distance between extensions such as
\begin{align*}
    \Y_{\bar\P}(\e)&=\left\{\y\in\Real^{d_{\y}}\::\:\p\tr\y\leq \pi(\p,\e),\:\forall \p\in \bar\P \right\}\\
    \hat\Y_{\bar\P}(\e)&=\left\{\y\in\Real^{d_{\y}}\::\:\p\tr\y\leq \hat\pi(\p,\e),\:\forall \p\in \bar\P \right\},
\end{align*}
where $\bar\P\subseteq \P$ is convex and compact. These sets nest the original sets (e.g. $\Y(\e)\subseteq\Y_{\bar\P}(\e)$) because the inequalities hold only for $\p\in\bar{\P}$, not for every $\p\in\P$. Moreover, the parts of the production possibility frontiers of the sets $\Y(\e)$ and $\Y_{\bar\P}(\e)$ coincide at points that are tangential to price vectors from $\bar\P$ (see Figure~\ref{fig: extended sets}).
\begin{figure}
\begin{center}
\begin{tikzpicture}[scale=0.6]
\draw[thick,->] (-1.2,4.5) -- (-0.62,6) node[anchor=north west] {$\p^{**}$};
\draw[thick,->] (2.4,1) -- (3.8,1.5) node[anchor=south east] {$\p^{*}$};
\draw[-,thick] (2.2,0) to (2,2) to [out=100,in=-10] (0,4) to (-2,4.2);
\draw [dashed] (2.8,0) -- (1.9,2.5);
\draw [dashed] (-2,4.8) -- (0.2,4);
\end{tikzpicture}
\end{center}
\caption{$\Y(\e)$ and $\Y_{\bar\P}(\e)$ for $d_{\y}=2$ and $\bar\P=\{\p\in\P\::\:\delta\leq\p_2/\p_1\leq 1/\delta,\, \norm{\p}\leq 1\}$, $0<\delta<1$. $\Y(\e)$ is the area under the solid curve. $\Y_{\bar\P}(\e)$ is the area under the dashed lines. Dashed lines correspond to two hyperplanes $\p^{*\prime}\y=\pi(\p^*,\e)$ and $\p^{**\prime}\y=\pi(\p^{**},\e)$. They are tangential to the solid curve. $\p^*$ is such that $\p^{*}_2/\p_1^{*}=\delta$ and $\p^{**}$ is such that $\p^{**}_2/\p_1^{**}=1/\delta$.}\label{fig: extended sets}
\end{figure}
\par
We now turn to the main result in this section, which establishes an equality relating the distance between $\hat{\pi}$ and $\pi$, and the distance between extensions of $\hat{\Y}$ and $\Y$. Our distance for these profit functions is given by
\[
\tilde{d}_{\bar\P}(\e) = \sup_{\p\in\bar\P}\norm{\dfrac{\hat\pi(\p,\e)-\pi(\p,\e)}{\norm{\p}}}\,.
\]
To state the following result, let $\mathcal{\bar\P}$ be a collection of all compact, convex, and nonempty subsets of $\P$.
\begin{thm}\label{thm:consistency}
Maintain the assumption that $\pi(\cdot,\e)$ is homogeneous of degree 1 and convex.\footnote{Recall that this is equivalent to price-taking, profit-maximizing behavior with technology described by a production correspondence.} Suppose, moreover, that for every $\e\in\E$, $\hat\pi(\cdot,\e)$ is an estimator of $\pi(\cdot,\e)$ that is homogeneous of degree $1$ and continuous. If $\hat\pi(\cdot,\e)$ is convex, then
\[
d_{H}(\Y_{\bar\P}(\e),\hat\Y_{\bar\P}(\e))= \tilde{d}_{\bar\P}(\e)\quad\as
\]
for every $\bar\P\in\mathcal{\bar\P}$. 
\end{thm}
\par
Theorem~\ref{thm:consistency} is a nontrivial extension of a well-known relation between the Hausdorff distance and the support functions of convex \emph{compact} sets to convex, closed, and \emph{unbounded} sets.\footnote{See \citet{kaido2014asymptotically} for a recent application of this result for convex compact sets.} Homogeneity of an estimator can be imposed by rescaling the data by dividing by one of the prices. Unfortunately, convexity can be more challenging to impose and so we turn to a related result that covers cases in which $\hat{\pi}$ is not convex. To formalize our result, we introduce two additional parameters:
\[
R_{\bar\P}(\e)=\sup_{\p\in\bar\P}\dfrac{\pi(\p,\e)}{\norm{\p}}\,,\quad r_{\bar\P}(\e)=\inf_{\p\in\bar\P}\dfrac{\pi(\p,\e)}{\norm{\p}}\,.
\]
\begin{prop} \label{prop:nonconvex}
Maintain the assumption that $\pi(\cdot,\e)$ is homogeneous and convex. Suppose, moreover, that for every $\e\in\E$, $\hat\pi(\cdot,\e)$ is an estimator of $\pi(\cdot,\e)$ that is homogeneous of degree $1$ and continuous. If $\tilde{d}_{\bar\P}(\e)=o_{p}(1)$ and $0<r_{\bar\P}(\e)<R_{\bar\P}(\e)<\infty$, then
\[
d_{H}(\Y_{\bar\P}(\e),\hat\Y_{\bar\P}(\e))\leq \tilde{d}_{\bar\P}(\e)\dfrac{R_{\bar\P}(\e)}{r_{\bar\P}(\e)}\dfrac{1+\tilde{d}_{\bar\P}(\e)/R_{\bar\P}(\e)}{1-\tilde{d}_{\bar\P}(\e)/r_{\bar\P}(\e)}
\]
with probability approaching $1$, for every $\bar\P\in\mathcal{\bar\P}$. In particular,
\[
d_{H}(\Y_{\bar\P}(\e),\hat\Y_{\bar\P}(\e)) = o_p(1)\,.
\]
\end{prop}

\section{Conclusion}\label{sec:conclusion}
In this paper we provide an update to classical duality theory in order to identify heterogeneous production sets in the presence of endogeneity, measurement error, omitted prices, and unobservable quantities. Our framework's main strength is to unpack rich heterogeneity as well as rich substitution/complementarity patterns with market level variation, using values of optimization problems. We achieve this by exploiting all shape constraints imposed by the economic environment we consider. This includes a key restriction that firms can be ranked in terms of productivity, and there are finitely many types of firms. Our identification results are constructive and can be applied in many available data sets.

\section*{Acknowledgments}

We thank the editor, the associate editor, and three anonymous referees for their comments and suggestions. We are grateful to Paul Grieco, Lance Lochner, Rosa Matzkin, Salvador Navarro, David Rivers, Susanne Schennach, Holger Sieg, and Al Slivinsky for useful comments and encouragement. We also thank the ceminar participants at Duke University, University of Montreal, McMaster University, and attendants of NASMES 2019, Empirical Microeconomics Workshop at University of Calgary, MEG 2019, CESG 2019, NAWMES 2020, vNAPW XI, WARP 2020, and CIREQ Montreal Econometrics Conference.

\bibliographystyle{econ}
\phantomsection\addcontentsline{toc}{section}{\refname}\bibliography{ref}
\appendix

\counterwithin{thm}{section}
\counterwithin{prop}{section}
\counterwithin{lem}{section}
\counterwithin{cor}{section}
\counterwithin{assumption}{section}
\counterwithin{defn}{section}
\counterwithin{rem}{section}

\setcounter{equation}{0}
\renewcommand{\theequation}{\thesection.\arabic{equation}}

\section{Proofs of Main Results}\label{app: proofs}
\subsection{Proof of Lemma~\ref{lemma: rich support}}
Fix $\yr^*$ and $\pu^*$. By homogeneity of degree 1 of the restricted profit function in prices and Assumption~\ref{assm:strict monot}, 
\[
\Delta\pi_r(\yr^*,\lambda\pu^*,\e)=\lambda\Delta\pi_r(\yr^*,\pu^*,\e)>0
\]
for every $e$ and $\lambda>0$. Since $\cup_{\lambda>0}\{\lambda\pu^*\}$ in the conditional support, we always can find $\lambda$ large enough and $\e^*$ such that Assumption~\ref{assm: measurement error}(ii) is satisfied.

\subsection{Proof of Theorem~\ref{prop: profit func discret I}}
First, note that since the support of $\rands{\eta}$ is a connected set (Assumption~\ref{assm: measurement error}(i)) and $\rand{e}$ is discrete, the conditional support of $\rands{\pi}_r$ conditional on $\rand{\y}_{-z}=\yr$ and $\rand{\p}_{z}=\pu$ is a union of connected sets for all $\yr$ and $\pu$ in their joint support. Hence, we can find the shortest (with respect to Lebesgue measure) isolated connected segment of the support for every $\yr$ and $\pu$. Next, among those short segments we can find the shortest one. By construction this segment will correspond to $(\yr^*,\pu^*,\e^*)$ from Assumption~\ref{assm: measurement error}(ii). As a result, under Assumption~\ref{assm: measurement error}, we can find an interval $[a,b]$ in the support of $\rands{\pi}_r$ conditional on  $\rand{\y}_{-z}=\yr^*$, $\rand{e}=\e^*$, and $\rand{\p}_{z}=\pu^*$ such that 
\[
\Prob{a\leq\pi_r(\yr^*,\pu^*,\e^*)+\rands{\eta}\leq b}=1
\]
and 
\[
\Prob{a\leq\pi_r(\yr^*,\pu^*,\e)+\rands{\eta}\leq b}=0
\]
for any $\e\neq\e^*$.
Hence, we identify 
\[
\pi_r(\yr^*,\pu^*,\e^*)=\Exp{\rands{\pi}_r|a\leq\rands{\pi}_r\leq b, \rand{\y}_{-z}=\yr^*, \rand{e}=\e^*, \rand{\p}_{z}=\pu^*},
\]
where we leverage that $\rands{\eta}$ has mean zero even after conditioning.

Thus, we can also recover the distribution of $\rands{\eta}$ by subtracting the identified $\pi_r(\yr^*,\pu^*,\e^*)$ from the known distribution of $\rands{\pi}_r|a\leq\rands{\pi}_r\leq b, \rand{\y}_{-z}=\yr^*, \rand{e}=\e^*, \rand{\p}_{z}=\pu^*$. Since $\rands{\eta}$ and $\pi_r(\rand{\y}_{-z},\rand{\p}_{z},\rand{\e})$ have bounded support and are independent conditional on $\rand{\y}_{-z}=\yr$ and $\rand{\p}_{z}=\pu$, we can constructively identify the moment generating function of  $\pi_r(\rand{\y}_{-z},\rand{\p}_{z},\rand{\e})$ conditional on $\rand{\y}_{-z}=\yr$ and $\rand{\p}_{z}=\pu$ as the ratio of the moment generating functions of $\rands{\pi}_r$ conditional on $\rand{\y}_{-z}=\yr$ and $\rand{\p}_{z}=\pu$  and $\rands{\eta}$. Since the distribution of $\pi_r(\rand{\y}_{-z},\rand{\p}_{z},\rand{\e})$ conditional on $\rand{\y}_{-z}=\yr$ and $\rand{\p}_{z}=\pu$ is discrete, its moment generating function is sufficient for its identification. Note that the moment generating function of $\rands{\eta}$ is well-defined and is never equal to zero since $\rands{\eta}$ is a bounded random variable. 
\par
Assumption~\ref{assm: monotone selection} implies that whenever a type $\e$ occurs with positive probability conditional on $\yr$ and $\pu$, then higher types also occur with positive probability. Assumption~\ref{assm:strict monot} then implies that the ranking over restricted profits is equivalent to the ranking over productivity $e$. As a result, if some firm of type $\e$ does not operate given $\yr$ and $\pu$, then it has to be a low type. Let $\Pi_r(\yr,\pu)$ be the support of $\pi_r(\rand{\y}_{-z},\rand{\p}_{z},\rand{\e})$ conditional on $\rand{\y}_{-z}=\yr$ and $\rand{\p}_{z}=\pu$. Fix some $\yr$ and $\pu$. Since the support of $\rand{\e}$ is finite, the set $\Pi_r(\yr,\pu)$ will also be finite. As a result, Assumption~\ref{assm:strict monot} implies that
\[
\pi_r(\yr,\pu,d_{e})=\max\left[\Pi_r(\yr,\pu)\right].
\]
That is, the most productive firm will make more profits than any other firm. Note that the firm with productivity $e=d_{e}-1$, if it is present in the market, will be the second one in terms of restricted profits : 
\[
\pi_r(\yr,\pu,d_{e}-1)
=\max\left[\Pi_r(\yr,\pu,s)\setminus\{\pi_r(\yr,\pu,d_{e})\}\right].
\]
In general, given $\yr$ and $\pu$, if the firm with productivity $e$ operates ($\abs{\Pi_r(\yr,\pu)}>d_{e}-e$), then
\[
\pi_r(\yr,\pu,e)
=\max\left[\Pi_r(\yr,\pu)\setminus\bigcup_{e'>e}\{\pi_r(\yr,\pu,e')\}\right].
\]
Note that we may not be able to identify the structural restricted profit function for arguments in which $\e$ is too low.

\subsection{Proof of Theorems~\ref{thm: attributes and profits}}
Fix some $x_{-1}$, and take $\y^*_{z}$ from the statement of the theorem and  $\e^{**}\in\E$ from condition (ii). We abuse notation and drop $\e^{**}$ and $\y^*_{-z}$. By homogeneity of degree 1 of $\pi_r(\cdot)$ we have that for every $\x$
\begin{align}\label{eq:1}
    \sum_{j=1}^{d_{\y_z}}\partial_{g_j}\pi_r(g(\x))g_{j}(\x_j)=\pi_r(g(\x))\,.
\end{align}
Moreover, since $\tilde\pi_r(\x)=\pi_r(g(\x))$ (recall that we dropped $\e^{**}$ and $\y^*_{-z}$ from the notation) and $\partial_{x_j}g_k(x_k)=0$ for $j\neq k$, we have that
\begin{align}\label{eq:2}
    \partial_{\x_j}\tilde\pi_r(\x)=\sum_{k}\partial_{g_k}\pi_r(g(\x))\partial_{x_j}g_{k}(\x_k)=\partial_{g_j}\pi_r(g(\x))\partial_{x_j}g_{j}(\x_j)\,,
\end{align}
for every $j=1,\dots,d_{\y_z}$.
Combining (\ref{eq:1}) and (\ref{eq:2}) we get that
\begin{align*}
     \sum_{j=1}^{d_{\y_z}}\partial_{\x_j}\tilde\pi_r(\x)\dfrac{1}{\partial_{\x_j}(\log(g_j(\x_j)))}=\tilde\pi_r(\x)
\end{align*}
as long as $0< \left| \dfrac{\partial_{\x_j}g_j(\x_j)}{g_j(\x_j)} \right| <\infty$ for every $j=1,\dots,d_{\y_z}$. This latter condition is satisfied for almost every $x_j$ with respect to Lebesgue measure by Assumption~\ref{assm:hedonic function}(iv). From Assumption~\ref{assm:hedonic function}(i), $g_1(x_1)=x_1$, so we obtain that
\begin{align}\label{eq:3}
     \sum_{j=2}^{d_{\y_z}}\partial_{\x_j}\tilde\pi_r(\x)\dfrac{1}{\partial_{\x_j}(\log(g_j(\x_j)))}=\tilde\pi_r(\x)-\partial_{\x_1}\tilde\pi_r(\x)x_1.
\end{align}
\par
Let $\tilde{t}=\left(\dfrac{1}{\partial_{\x_j}(\log(g_j(\x_j)))}\right)_{j=2,\dots,d_{\y_z}}$. Note that $\tilde{t}$ does not depend on $x_{1}$. Since $\tilde\pi_r$ satisfies the rank condition there exists a nonsingular $A(\tilde\pi_r(x^*))$ and $b$ such that equation (\ref{eq:3}) can be rewritten as
\begin{align}
    A\tilde{t}=b\,,
\end{align}
where $b=(b_l)_{l=1,\dots,d_{\y_z}-1}$ and $b_l=\tilde\pi_r(\x^*_{l})-\partial_{\x_{1}}\tilde\pi_r(\x^*_{l})t_l$. Since $A(\tilde\pi_r(x^*))$ is of full rank and is identified, and $b$ is identified, $\tilde{t}$ is identified. Since the choice of $x_{-1}$ was arbitrary and we know the location (Assumption~\ref{assm:hedonic function}(ii)), we identify $g_{j}(\cdot)$ for every $j=1,\dots,d_{\y_z}$.

\subsection{Proof of Theorem~\ref{thm:identification of Y}}
It is immediate that $\tilde\Y(\e)$ is closed, convex, and satisfies free disposal for every $\e\in\E$. Moreover, $\max_{\y\in\tilde\Y(\e)}\p\tr\y=\pi(\p,\e)$ for every $\p\in\P(\e)$ and $\e\in\E$. Thus, conclusion (i) follows from the fact that $\pi(\p,\e)$ is identified for each $\p \in \P(\e)$ and $\e\in\E$ by Theorem~\ref{prop: profit func discret I}. 
\par
To establish conclusion (ii), recall that under the assumptions of Theorem~\ref{prop: profit func discret I}, any given production set $\Y'(\e)$ can generate the data if and only if $\max_{\y\in\Y'(\e)}\p\tr\y=\pi(\p,\e)$ for every $\p\in\P(\e)$. The set $\tilde{\Y}(\e)$ is constructed as the largest set (not necessary production set) consistent with profit maximization. This set is closed, convex, and satisfies free disposal. Since a production correspondence \textit{also} must satisfy the recession cone property, we obtain that $\Y'(\e) \subseteq \tilde{\Y}(\e)$.
\par
To prove (iii), note that since $\pi(\cdot,\e)$ is homogeneous of degree $1$ for every $\e\in\E$ we can identify $\pi(\cdot,\e)$ over 
\[
\bigcup_{\lambda > 0 }\left\{\lambda \p\::\:\p\in\P(\e)\right\}\,.
\] 
Next, since $\pi(\cdot,\e)$ is convex it is continuous, hence it is identified over 
\[
\mathrm{int}\left(\mathrm{cl}\left(\bigcup_{\lambda > 0 }\left\{\lambda \p\::\:\p\in\P(\e)\right\}\right)\right)\,.
\]
When Assumption~\ref{assm:rich price} holds, identification of $\Y(\cdot)$ follows from Corollary 9.18 in \citet{kreps2012}.

\subsection{Proof of Proposition~\ref{prop: identified y}}
Fix some $\e\in\E$. To simplify notation we drop $\e$ from the objects below (e.g. $\pi(\p,\e)=\pi(\p)$ and $\y_p(\e)=\y_p$).
Suppose $\{\y_p\}_{\p\in\P}$ can generate $\{\pi(\p)\}_{\p\in\P}$. Since  $\{\y_p\}_{\p\in\P}$ are profit-maximizing output/input vectors we must have $\p\tr\y_{\p}=\pi(\p)$. To prove that $\p^{* \prime}\y_{\p^{* \prime}}\geq \p^{* \prime}\y_{p}$ for all $\p,\p^{*}\in\P$, assume the contrary. But then $\y_{\p_{*}}$ is not maximizing profits at $\p^{*}$ since $\y_p$ is available. The contradiction proves necessity.
\par
To prove sufficiency consider 
\[
\Y^*=\hull(\{\y_p\}_{\p\in\P})+\Real^{d_\y}_{-}\,,
\]
where $\hull(A)$ denotes the convex hull of a set $A$, i.e. the smallest convex set containing $A$. The summation is the Minkowski sum. $\Y^*$ is sometimes referred to as the free-disposal convex hull of $\{\y_p\}_{\p\in\P}$. In particular, note that $\Y^*$ is convex, closed, and satisfies free disposal.

We obtain that for every $p\in\Real^{d_\y}_{++}\cap\mathbb{S}^{d_\y-1}$,
\[
\sup_{\y\in\Y^*}\p\tr\y=\sup_{\y\in\hull(\{\y_p\}_{\p\in\P})}\p\tr\y+\sup_{\y\in\Real^{d_{\y}}_{-}}\p\tr\y=\sup_{\y\in\hull(\{\y_p\}_{\p\in\P})}\p\tr\y\,.
\]
Because $\P$ is finite, $\{\y_p\}_{\p\in\P}$ is bounded. Thus, its convex hull $\hull(\{\y_p\}_{\p\in\P})$ is also bounded. This implies that $\sup_{\y\in\Y'}\p\tr\y$ is finite for every $p\in\Real^{d_\y}_{++}\cap\mathbb{S}^{d_\y-1}$, hence the recession cone property is satisfied for the set $\Y^*$.\footnote{We note that \citet{varian1984nonparametric} studies a result related to this proposition, taking as primitives a deterministic dataset of prices and quantities. He does not verify the recession cone property.} 
\par
It is left to show that 
\[
\pi(\p,\e)=\p\tr\y_{\p}=\sup_{\y\in\Y^*}\p\tr\y
\]
for every $\p\in\P\cap\mathbb{S}^{d_\y-1}$. The first equality is assumed. Suppose the second equality is not true for some $\p^{*}$. Then there exists $\tilde\y\in\Y^*$ such that $\p^{* \prime}\y_{\p^{*}}<\p^{* \prime}\tilde\y$. Since $\tilde\y\in\Y^*$ it can be represented as a finite convex combination of points from $\{\y_\p\}_{\p\in\P}$. But since
\[
\p^{* \prime} \y_{\p^{*}}\geq\p^{* \prime} \y_\p\,,
\]
for all $\p,\p^{*}\in\P$ it has to be the case that 
\[
\p^{* \prime}\y_{\p^{*}}\geq \p^{* \prime}\tilde\y.
\]
The contradiction completes the proof. Since the choice of $\e$ was arbitrary the result holds for all $\e\in\E$.
\subsection{Proof of Theorem~\ref{thm:consistency} and Proposition~\ref{prop:nonconvex}}
The Hausdorff distance between two convex sets $A, B \subseteq \Real^{d_\y}$ is given by
\[
d_H (A,B) = \max \left\{ \sup_{a \in A } \inf_{b \in B} \| a - b \|, \sup_{b \in B } \inf_{a \in A} \| a - b \| \right\}\,.
\]
Alternatively, the Hausdorff distance can be defined as 
\[
d_H (A,B) = \inf\{\rho\geq 0\::\:A\subseteq B+\rho \mathbb{B}^{d_\y-1},B\subseteq A+\rho \mathbb{B}^{d_\y-1}\}\,,
\]
where $\mathbb{B}^{d_\y-1}=\{\y\in\Real^{d_\y}\::\:\norm{\y}\leq 1\}$ is the unit ball and $\inf\{\emptyset\}=\infty$.
The support function of a closed convex set $A$ is defined for $u \in \mathbb{R}^{d_\y}$ via $h_A(u) = \sup_{w \in A} u'w$. If $A$ is unbounded in direction $u$, then $h_A(u)=\infty$.
\par
As preparation, we need a technical lemma. This lemma involves a polar cone, which for a set $C$ is defined by
\[
\mathrm{PolCon}(C) = \{u\in\Real^{d_\y}\::\:u\tr\p\leq 0,\,\forall\p\in C\}.
\]
\begin{lem} \label{lem:polar}
Let $\bar\P\subseteq\mathbb{S}^{d_\y-1}$ be a closed set such that $\cup_{\lambda>0}\{\lambda\p,\p\in\bar\P\}$ is a closed, convex cone, and let $a:\Real^{d_\y}\to\Real$ be a convex, homogeneous of degree 1 function. Define
\[
A=\{\y\in\Real^{d_\y}\::\:\p\tr\y\leq a(\p),\,\forall\p\in\bar\P\}.
\]
If $\mathrm{PolCon}(\bar\P)$ is nonempty, then for any $u\in \mathbb{S}^{d_\y-1}$,
\[
h_{A}(u)=\begin{cases}
a(u),& \text{if } u\in\bar\P,\\
+\infty,& \text{otherwise}.
\end{cases}
\]
\end{lem}
\begin{proof}
Case 1. Take $u\in\bar\P$. Since $a(\cdot)$ is convex and homogeneous of degree 1 $h_A(u)=a(u)$.
\par
Case 2. Take $u\in\mathbb{S}^{d_\y-1}\setminus\bar\P$. First, we establish that there always exists $u^*\in\mathrm{PolCon}(\bar\P)$ such that $u\tr u^*>0$. To prove this suppose to the contrary that for every $u^*\in\mathrm{PolCon}(\bar\P)$, $u\tr u^*\leq 0$, it follows that $u\in \mathrm{PolCon}(\mathrm{PolCon}(\bar\P))$. The latter is not possible, since $\mathrm{PolCon}(\mathrm{PolCon}(\bar\P))$ is the smallest closed convex cone containing $\bar\P$ (\citealp{rock70}, Theorem 14.1), and $u\not\in\bar\P$ by assumption.

For some $u^*$ that satisfies $u\tr u^*>0$, consider $\y^m=\y^0+mu^*$, $m=1,2,\dots$, where $\y^0$ is an arbitrary point from $A$. Since $u^*\in\mathrm{PolCon}(\bar\P)$, by construction $u^{* \prime} p \leq 0$ for all $\p\in\bar\P$. Using this fact, note that $\y^m\in A$ for all $m=1,2,\dots$ since
\[
\p\tr\y^m=\p\tr\y^0+m u^{*\prime}\p\leq a(\p)+0
\]
for all $\p\in\bar\P$. Finally,
\[
h_A(u)\geq u\tr\y^m=u\tr\y^0+mu\tr u^*
\]
diverges to $+\infty$, since $u\tr u^*>0$.
\end{proof}
\par
We now provide a key lemma. This result generalizes a classical result that holds for $\bar\P = \mathbb{S}^{d_\y-1}$. To our knowledge this result is new, and it may be of independent interest.
\begin{lem}\label{lem:truncation}
Let $d_{\y}\geq 2$ and let the functions $a,b:\Real^{d_{\y}}_{++}\to\Real$ be convex and homogeneous of degree $1$. Define
\begin{align*}
    A&=\left\{\y\in\Real^{d_{\y}}\::\:\p\tr\y\leq a(\p),\:\forall \p\in \bar\P \right\}\,,\\
    B&=\left\{\y\in\Real^{d_{\y}}\::\:\p\tr\y\leq b(\p),\:\forall \p\in \bar\P \right\}\,,
\end{align*}
where $\bar\P\subseteq\Real^{d_\y}_{++}$ is convex and compact. Then
\[
d_H(A,B)=\sup_{\p\in\bar\P}\norm{a(\p/\norm{\p})-b(\p/\norm{\p})}\,.
\]
\end{lem}
\begin{proof}
For closed convex sets $C,D\subseteq\Real^{d_\y}$ the following is true: $C\subseteq D$ if and only if $h_{C}(u)\leq h_{D}(u)$ for all $u\in\mathbb{S}^{d_\y-1}$. Hence,
\begin{align*}
&\{\rho\in\Real_{+}\::\:A\subseteq B+\rho\mathbb{B}^{d_\y-1},B\subseteq A+\rho\mathbb{B}^{d_\y-1}\} \iff\\ &\{\rho\in\Real_{+}\::\:h_A(u)\leq h_{B+\rho\mathbb{B}^{d_\y-1}}(u),h_B(u)\leq h_{A+\rho\mathbb{B}^{d_\y-1}}(u),\forall u\in\mathbb{S}^{d_\y-1}\}\,.
\end{align*}
Because $\bar{P}$ is a subset of $\Real^{d_\y}_{++}$, its polar cone $\mathrm{PolCon}(\overline{P})$ is nonempty; in particular the polar cone contains the negative unit vector $(-1,\ldots, -1)\tr$. The set $\bar{P}$ satisfies the conditions of Lemma~\ref{lem:polar}, and so we obtain that $h_A(u)=h_{B+\rho\mathbb{B}^{d_\y-1}}(u)=h_B(u)=h_{A+\rho\mathbb{B}^{d_\y-1}}(u)=\infty$ for all $u\in\mathbb{S}^{d_\y-1}\setminus\{p/\norm{p}\,,\,\p\in\bar\P\}$. Hence,  
\begin{align*}
\{\rho\in\Real_{+}& \::\:A\subseteq B+\rho\mathbb{B}^{d_\y-1},B\subseteq A+\rho\mathbb{B}^{d_\y-1}\} \\
& = \{\rho\in\Real_{+}\::\:\:h_A(u)\leq h_{B+\rho\mathbb{B}^{d_\y-1}}(u), \\
& \qquad \qquad h_B(u)\leq h_{A+\rho\mathbb{B}^{d_\y-1}}(u),\forall u\in\{p/\norm{p}\::\p\in\bar\P\}\}\\
& = \{\rho\in\Real_{+}\::\:\:h_A(u)\leq h_{B}(u)+h_{\rho\mathbb{B}^{d_\y-1}}(u), \\
& \qquad \qquad h_B(u)\leq h_{A}(u)+h_{\rho\mathbb{B}^{d_\y-1}}(u),\forall u\in\{p/\norm{p}\::\:\p\in\bar\P\}\}\\
& = \{\rho\in\Real_{+}\::\:h_A(u)\leq h_{B}(u)+\rho,h_B(u)\leq h_{A}(u)+\rho,\forall u\in\{p/\norm{p}\::\p\in\bar\P\}\}\\
& = \{\rho\in\Real_{+}\::\:\sup_{u\in\{p/\norm{p}\::\:\p\in\bar\P\}}\norm{h_A(u)-h_{B}(u)}\leq \rho\}\,.
\end{align*}
Now note that $a(\p)$ and $b(\p)$ are values of the support functions of $A$ and $B$ evaluated at $p\in\bar\P$, respectively, since $a(\cdot)$ and $b(\cdot)$ are homogeneous of degree $1$ and convex. Thus,
\[d_{H}(A,B)=\sup_{\p\in\bar\P}\norm{a(\p/\norm{\p})-b(\p/\norm{\p})}.
\]
\end{proof}
To prove Theorem~\ref{thm:consistency} note that since $\pi(\cdot,\e)$ and $\hat\pi(\cdot,\e)$ are homogeneous of degree $1$, we have
\begin{align*}
    \pi(\p,\e)/\norm{\p}&=\pi\left(\p/\norm{\p},\e\right)\,,\\
    \hat\pi(\p,\e)/\norm{\p}&=\hat\pi\left(\p/\norm{\p},\e\right)\,,
\end{align*}
for all $p\in\bar\P$ and $\e\in\E$. Thus, Theorem~\ref{thm:consistency} is obtained as corollary.

We now turn to the proof of Proposition~\ref{prop:nonconvex}. We first present two lemmas, which are modifications of Lemmas 6 and 7 in \citet{brunel2016concentration}.
\begin{lem}\label{lem:KKT}
Assume that $\bar\P\subseteq\mathbb{S}^{d_\y-1}\bigcap\P$ is compact and $\cup_{\lambda>0}\{\lambda\p\::\:\p\in\bar\P\}$ is convex.
Let $a:\bar\P\to\Real$ be a continuous function. Let $A=\{\y\in\Real^{d_\y}\::\:\p\tr\y\leq a(\p),\,p\in\bar\P\}$ be nonempty. It follows that for all $\p^*\in\bar\P$ there exists $\y^*\in A$ such that $h_{A}(\p^*)=\p^{*\prime}\y^*$. Moreover, there exists $\P^*\subseteq\bar\P$ such that
\begin{enumerate}
    \item The cardinality of $\P^*$ is less than or equal to $d_{\y}$;
    \item $\p\tr\y^*=a(\p)$ for all $p\in\P^*$;
    \item $\p^*=\sum_{\p\in\P^*}\lambda_{\p}\p$ for some nonnegative numbers $\lambda_{p}$.
\end{enumerate}
\end{lem}
\begin{proof}
Fix some $\p^*\in\bar\P$. Note that $h_{A}(\p^*)\leq a(\p^*)<\infty$. Since $A$ is closed, by the supporting hyperplane theorem $h_{A}(\p^*)=\p^{*\prime}\y^*$ for some $\y^*\in A$.
\par
The rest of the lemma follows from Theorem 2(b) in \citet{lopez2007semi} if we show that $\P'=\{\p\in\bar\P\::\:\p\tr\y^*=a(\p)\}$ is nonempty. By way of contradiction assume that $\P'$ is empty. Hence,  $\p\tr\y^*<a(\p)$ for all $p\in\bar\P$. Since the function $a(\cdot)-\cdot\tr\y^*$ is strictly positive on a compact $\bar\P$, there exists $\nu>0$ that bounds $a(\cdot)-\cdot\tr\y^*$ from below. Hence, for every $\p\in\bar\P$,
\begin{align*}
    \p\tr(\y^*+\nu\p^*)=\p\tr\y^*+\nu\p\tr\p^*\leq a(\p)-\nu+\nu\p\tr\p^*\leq a(\p)\,.
\end{align*}
Thus, $(\y^*+\nu\p^*)\in A$. But the later is not possible since $\p^*(\y^*+\nu\p^*)=a(\p^*)+\nu>a(\p^*)$ implies that $\y^*$ is not a maximizer. Thus, $P'$ is nonempty.
\end{proof}
\begin{lem}\label{lem:support func converg}
Assume that $\bar\P\subseteq\mathbb{S}^{d_\y-1}\bigcap\P$ is compact and $\cup_{\lambda>0}\{\lambda\p\::\:\p\in\bar\P\}$ is convex.
Let $a:\bar\P\to\Real$ be continuous convex homogeneous of degree $1$ function and $\{b_n:\bar\P\to\Real\}$ be a sequence of continuous homogeneous of degree $1$ functions such that 
\begin{align*}
    A&=\left\{\y\in\Real^{d_{\y}}\::\:\p\tr\y\leq a(\p),\:\forall \p\in \bar\P \right\}\,,\\
    B_n&=\left\{\y\in\Real^{d_{\y}}\::\:\p\tr\y\leq b_n(\p),\:\forall \p\in \bar\P \right\}\,,
\end{align*}
are nonempty for all $n\in\Natural$. Assume that $\eta_n=\sup_{\p\in\bar\P}\norm{a(\p)-b_n(\p)}=o(1)$ and $0<r=\inf_{\p\in\bar\P} a(\p)<R=\sup_{\p\in\bar\P}a(\p)<\infty$. Then there exists $N>0$ such that 
\[
\sup_{\p\in\bar\P}\norm{a(\p)-h_{B_n}(\p)}\leq \tilde{d}_{n}\dfrac{R}{r}\dfrac{1+\eta_n/R}{1-\eta_n/r}
\]
for all $n>N$.
\end{lem}
\begin{proof}
Fix some $\p^*\in\bar\P$ and some $n$ such that $\eta_n<r$. By Lemma \ref{lem:KKT} there exists a finite set $\P_n^*$, a collection of nonnegative numbers $\{\lambda_{p,n}\}_{p\in\P_n^*}$ and $\y^*_n\in B_{n}$ such that $h_{B_n}=\p^{*\prime}\y^*_n$, $\p^*=\sum_{p\in\P_n^*}\lambda_{p,n}p$, and $\p\tr\y^*_n=b_n(\p)$ for all $p\in\P^*_n$. Note that for all $p\in\p_n^*$ we have that $b_n(\p)=h_{B_n}(\p)$. Then 
\begin{align}
    a(\p^*)&=h_{A}(\p^*)=h_{A}\left(\sum_{\p\in\P_n^*}\lambda_{\p,n}\p\right)\leq\sum_{\p\in\P_n^*}\lambda_{\p,n}h_{A}(\p)= 
    \sum_{\p\in\P_n^*}\lambda_{\p,n}a(\p)\leq \sum_{\p\in\P_n^*}\lambda_{\p,n}(b_n(\p)+\eta_n)\\
    \nonumber&=\sum_{\p\in\P_n^*}\lambda_{\p,n}\p\tr\y^*_n+\eta_n\sum_{\p\in\P_n^*}\lambda_{\p,n}=\p^{*\prime}\y^*_n+\eta_n\sum_{\p\in\P_n^*}\lambda_{\p,n}=h_{B_n}(\p^*)+\eta_n\sum_{\p\in\P_n^*}\lambda_{\p,n}\,.
\end{align}
Moreover,
\begin{align}\label{eq:support func conv 2}
    h_{B_n}(\p^*)\leq b_n(\p^*)\leq a(\p^*)+\eta_n\,.
\end{align}
Hence, $\norm{a(\p^*)-h_{B_n}(\p^*)}\leq \eta_n\max\{1,\sum_{\p\in\P_n^*}\lambda_{\p,n}\}$.
\par
Next note that the inequality in (\ref{eq:support func conv 2}) implies that
\[
\sum_{\p\in\P^*_n}\lambda_{\p,n}\p\tr\y^*_n=\p^{*\prime}\y^*_n=h_{B_n}(\p^*)\leq a(\p^*)+\eta\leq R+\eta_n\,.
\]
In addition, 
\[
\sum_{\p\in\P^*_n}\lambda_{\p,n}\p\tr\y^*_n=\sum_{\p\in\P^*_n}\lambda_{\p,n}b_n(\p)\geq \sum_{\p\in\P^*_n}\lambda_{\p,n}(a(\p)-\eta_n)\geq \sum_{\p\in\P^*_n}\lambda_{\p,n}(r-\eta_n)\,.
\]
Hence,
\[
\sum_{\p\in\P^*_n}\lambda_{\p,n}\leq \dfrac{R+\eta_n}{r-\eta_n}\,.
\]
As a result,
\[
\norm{a(\p^*)-h_{B_n}(\p^*)}\leq \eta_n\max \left\{1,\sum_{\p\in\P_n^*}\lambda_{\p,n} \right\}=\eta_n\max\left\{1,\dfrac{R+\eta_n}{r-\eta_n}\right\}=\eta_n\dfrac{R}{r}\dfrac{1+\eta_n/R}{1-\eta_n/r}\,.
\]
\end{proof}
To prove Theorem~\ref{thm:consistency} note that since $\pi(\cdot,\e)$ and $\hat\pi(\cdot,\e)$ are homogeneous of degree $1$, we have
\begin{align*}
    \pi(\p,\e)/\norm{\p}&=\pi\left(\p/\norm{\p},\e\right)\,,\\
    \hat\pi(\p,\e)/\norm{\p}&=\hat\pi\left(\p/\norm{\p},\e\right)\,.
\end{align*}
To prove Proposition~\ref{prop:nonconvex}, note that by Lemma~\ref{lem:truncation}, with probability $1$,
\[
d_{H}(\Y_{\bar\P}(\e),\hat\Y_{\bar\P}(\e))=\sup_{\p\in\bar\P}\norm{\pi(\p/\norm{\p},\e)-h_{\hat\Y_{\bar\P}(\e)}(\p/\norm{\p})}\,. 
\]
The conclusion then follows by applying Lemma~\ref{lem:support func converg} to the right hand side of the equality above. 

\section{Additional Details on Estimation}\label{app: estimation and simulation}

This section presents an estimator that is used in the illustrative empirical application in Appendix~\ref{app: application}. The estimator builds on the constructive identification result of Theorem~\ref{prop: profit func discret I} and applies with continuous measurement error and discrete heterogeneity in productivity. It proceeds in two steps. First, we find a ``minimal-width'' region of profits that is used to estimate the distribution of measurement error. This uses the well-separatedness structure of Theorem~\ref{prop: profit func discret I}. Second, we use the estimate of the distribution of measurement error to estimate the distribution of structural profit.

\subsection{Estimation of Restricted Profit Function}\label{app: estimation of profit}
To simplify the exposition, in this section we assume the researcher observes data on (unrestricted) profits and prices from $M$ markets $\{\rands{\pi}_{i,m},\rand{p}_{m}\}_{i=1,\dots,N; m=1,\dots, M}$. Here, $\rands{\pi}_{i,m}$ is the observed profit of firm $i$ in market $m$, which may be mismeasured. The index $i$ can be market specific, so in particular firm $1$ in market $1$ may differ from firm $1$ in market $2$. For each market $m$, all firms face the same price vector $\rand{p}_{m}$. There are $N$ firms in every market.\footnote{We assume the same number of firms in every market only to simplify the exposition.} The general case with restricted profits can be handled similarly. We assume that $\rand{p}_{m}=\rand{p}_{m'}$ with probability $1$ if and only if $m=m'$ (i.e., markets have different prices). We require the number of firms per market $N$ to grow to infinity. The number of markets $M$ can be fixed, grow to a finite constant, or diverge to infinity as long as it grows slower than $N$.
\subsubsection{Estimation of the Measurement Error Distribution}
\par
With this setup, we can estimate the distribution of measurement error. We do so by first finding a partition of profits in which some region has ``minimal width.'' To formalize this we first describe how we partition. For a finite set of distinct reals $\mathcal{T}=\{t_{\ell}\}_{{\ell}=1}^L$ and $\kappa>0$, let $t^{({\ell})}$ be the ${\ell}$-th smallest element of $\mathcal{T}$. Next, let $\{T^{\kappa}_k\}_{k=1}^{K_{\kappa}}$ be a smallest (in terms of cardinality) partition of $\mathcal{T}$ such that $\max T^{\kappa}_k<\min T^{\kappa}_{k+1}$ for all $k=1,\dots,K_{\kappa}-1$, and $\abs{t^{(\ell)}-t^{(j)}}\leq \abs{\ell-j}\kappa$ for any $T^{\kappa}_{k}$ and any $t^{({\ell})},t^{(j)}\in T^{\kappa}_{k}$. Such partition always exists but may not be unique. For our purposes, any such partition works.\footnote{This partition is related to so-called density-based clustering. See \citet{kriegel2011density} for a review.} Let $d(T^{\kappa}_k)=\left(\max T^{\kappa}_k-\min T^{\kappa}_k\right)$ be the diameter of the set $T^{\kappa}_k$.  Given the partition, let
$k^{*}$ be the smallest integer such that $d(T^{\kappa}_{k^*})\leq d(T^{\kappa}_k)$ for each $k = 1, \ldots, K_{\kappa}$. That is, $T_{k^*}^{\kappa}$ is the  \emph{first-shortest} element of the partition (see Figure~\ref{fig: construction of C} for an example). Finally, let 
\[
\mathcal{C}(\mathcal{T},\kappa)=\left\{t-\dfrac{1}{\abs{T_{k^*}^{\kappa}}}\sum_{t'\in T_{k^*}^{\kappa}}t'\right\}_{t\in T_{k^*}^{\kappa}}.
\]
That is, the operator $\mathcal{C}$ takes the set $\mathcal{T}$ and threshold $\kappa>0$, computes the set $T_{k^*}^{\kappa}$, and then re-centers this set such that the sample average of elements of it is zero. That is, 
\[
\dfrac{1}{\abs{\mathcal{C}(\mathcal{T},\kappa)}}\sum_{t\in \mathcal{C}(\mathcal{T},\kappa)}t=0.
\]
\begin{figure}
\begin{center}
\begin{tikzpicture}[scale=0.9]
\draw [fill=black] (-6,0) circle[radius=.1];
\draw [fill=black] (-5.89,0) circle[radius=.1];
\draw [fill=black] (-5.69,0) circle[radius=.1];
\draw [fill=black] (-4.97,0) circle[radius=.1];
\draw [fill=black] (-4.24,0) circle[radius=.1];
\draw [fill=black] (-4.23,0) circle[radius=.1];
\draw [fill=black] (-3.54,0) circle[radius=.1];
\draw [fill=black] (-3.29,0) circle[radius=.1];

\draw [fill=black] (-1,0) circle[radius=.1];
\draw [fill=black] (-0.7,0) circle[radius=.1];
\draw [fill=black] (-0.6,0) circle[radius=.1];
\draw [fill=black] (-0.4,0) circle[radius=.1];

\draw [fill=black] (3,0) circle[radius=.1];
\draw [fill=black] (3.89,0) circle[radius=.1];
\draw [fill=black] (3.69,0) circle[radius=.1];
\draw [fill=black] (4.97,0) circle[radius=.1];
\draw [fill=black] (4.24,0) circle[radius=.1];
\draw [fill=black] (4.23,0) circle[radius=.1];
\draw [fill=black] (5.54,0) circle[radius=.1];
\draw [fill=black] (5.29,0) circle[radius=.1];

\node at (-4.5,-0.6) {$T^{\kappa}_1$};
\node at (-0.75,-0.6) {$T^{\kappa}_2$};
\node at (4.25,-0.6) {$T^{\kappa}_3$};
\end{tikzpicture}
\end{center}
\caption{Partitioning into 3 sets. $T^{\kappa}_{2}$ is the shortest element of the partition ($k^*=2$). }\label{fig: construction of C}
\end{figure}

\par
Given a sequence of positive reals $\kappa_N$ that slowly converges to $0$, let $m_N^*$ be a market that has the smallest  $d\left(\mathcal{C}\left(\{\rands{\pi}_{i,m}\}_{i=1}^N,\kappa_N)\right)\right)$. Then, under the assumptions of Theorem~\ref{prop: profit func discret I}, the elements of $\mathcal{C}\left(\{\rands{\pi}_{i,m_N^*}\}_{i=1}^N,\kappa_N)\right)$ mimic the unobserved realizations of the measurement error. Thus, we can apply any consistent estimator (e.g., kernels or sieves) to $\mathcal{C}\left(\{\rands{\pi}_{i,m_N^*}\}_{i=1}^N,\kappa_N)\right)$ to obtain a consistent estimator of the p.d.f. of the measurement error. 
\begin{proposition}
Take $\kappa_N$ such that $\kappa_N=o(1)$ and $\log(N)/(N\kappa_N)=o(1)$. Assume the assumptions of Theorem~\ref{prop: profit func discret I} are satisfied. Assume $\rands{\eta}$ admits a continuous p.d.f. $f_{\rands{\eta}}$. Suppose there is an estimator $\hat{f}_{\rands{\eta}}(\cdot, \{\rands{\eta}_i\})$ that is consistent for $f_{\rands{\eta}}$, based on an i.i.d. sample from $f_{\rands{\eta}}$, denoted $\{\rands{\eta}_i\}$. Let $m^*_N$ be such that $d\left(\mathcal{C}(\{\rands{\pi}_{i,m_N^*}\},\kappa_N)\right)\leq d\left(\mathcal{C}(\{\rands{\pi}_{i,m}\},\kappa)\right)$ for all $m$ with probability $1$. It follows that $\hat{f}_{\rands{\eta}}\left(\cdot, \mathcal{C}(\{\rands{\pi}_{i,m_N^*}\},\kappa_N)\right)$ is a consistent estimator of $f_{\rands{\eta}}$.
\end{proposition}
\begin{proof}
First, note that for any $\kappa>0$ and two random variables $\rands{\eta}_1$ and $\rands{\eta_2}$ that are independently and identically distributed according to $f_{\rands{\eta}}$,
\begin{align*}
    &p(\kappa)=\Prob{\abs{\rands{\eta}_1-\rands{\eta}_2}\leq\kappa}=\int_{-K_1}^{-K_1+\kappa}[F_{\rands{\eta}}(x+\kappa)-F_{\rands{\eta}}(-K_1)]f_{\rands{\eta}}(x)dx+\\
    &\int_{K_2-\kappa}^{K_2}[F_{\rands{\eta}}(K_2)-F_{\rands{\eta}}(x-\kappa)]f_{\rands{\eta}}(x)dx+\int_{-K_1+\kappa}^{K_2-\kappa}[F_{\rands{\eta}}(x+\kappa)-F_{\rands{\eta}}(x-\kappa)]f_{\rands{\eta}}(x)dx,
\end{align*}
where $F_{\rands{\eta}}$ is the c.d.f. of $\rands{\eta}$ supported on $[-K_1,K_2]$, where $K_1,K_2>0$.
The first term in the above equation can be bounded by
\begin{align*}
\int_{-K_1}^{-K_1+\kappa}[F_{\rands{\eta}}(x+\kappa)-F_{\rands{\eta}}(-K_1)]f_{\rands{\eta}}(x)dx & \leq \max_x f_{\rands{\eta}}(x) \int_{-K_1}^{-K_1 + \kappa} [F_{\rands{\eta}}(-K_1+2\kappa)-F_{\rands{\eta}}(-K_1)] dx \\
& \leq \max_x f_{\rands{\eta}}(x)\dfrac{F_{\rands{\eta}}(-K_1+2\kappa)-F_{\rands{\eta}}(-K_1)}{2\kappa}2\kappa^2.
\end{align*}
Similarly, the second term is bounded above by
\[
\max_x f_{\rands{\eta}}(x) \dfrac{F_{\rands{\eta}}(K_2)-F_{\rands{\eta}}(K_2-2\kappa)}{2\kappa}2\kappa^2.
\]
As a result, since $\max_x f_{\rands{\eta}}(x)<\infty$ ($f_{\rands{\eta}}$ is continuous on a compact support) and $F_{\rands{\eta}}$ has a bounded and continuous derivative on a compact set, as $\kappa\to 0$,
\begin{align*}
    &p(\kappa)=2\kappa\int_{-K_1+\kappa}^{K_2-\kappa}[f_{\rands{\eta}}(x)+O(\kappa)]f_{\rands{\eta}}(x)dx+O(\kappa^2)
\end{align*}
and 
\[
\lim_{\kappa\to 0}\dfrac{p(\kappa)}{\kappa}=C=2\int_{-K_1}^{K_2}f^2_{\rands{\eta}}(x)dx>0.
\]
Second, note that given an i.i.d. sample $\{\rands{\eta}_i\}_{i=1}^n$ from $f_{\rands{\eta}}$
\begin{align*}
&\Prob{\max_{i}\min_{j\neq i}\abs{\rands{\eta}_i-\rands{\eta}_j}\leq \kappa}=\Prob{ \bigcap_{i = 1}^n \left\{ \min_{j\neq i}\abs{\rands{\eta}_i-\rands{\eta}_j}\leq \kappa \right\}}\\
&\geq \sum_{i=1}^n\Prob{\min_{j\neq i}\abs{\rands{\eta}_i-\rands{\eta}_j}\leq \kappa} -(n-1)=1-\sum_{i=1}^n\Prob{\min_{j\neq i}\abs{\rands{\eta}_i-\rands{\eta}_j}>\kappa}=\\
&=1-n\Prob{\abs{\rands{\eta}_1-\rands{\eta}_2}>\kappa}^{n-1}=1-n(1-p(\kappa))^{n-1}=1-n(1-C\cdot\kappa+o(\kappa))^{n-1}.
\end{align*}
Hence, 
\begin{align*}
&\lim_{N\to\infty}\Prob{\max_{i}\min_{j\neq i}\abs{\rands{\eta}_i-\rands{\eta}_j}\leq \kappa_N}\geq 1-\lim_{N\to\infty}N\exp(-C(N-1)\kappa_N)=1,
\end{align*}
where the last equality follows from the fact that $\kappa_N$ converges to $0$ slower that $\log(N)/N$.

This bound on the measurement error distribution implies that the largest distance between neighboring observations that are coming from the same productivity level becomes less than $\kappa_N$ with probability approaching 1 as $N$ increases. Hence, since $\kappa_N$ converges to zero and we pick the shortest element of the partition, $\mathcal{C}(\{\rands{\pi}_{i,m_N^*}\},\kappa_N)$ will contain i.i.d observations that correspond to the same productivity level with probability approaching 1. Thus, any consistent estimator that is based on an i.i.d. sample will be consistent. 
\end{proof}
\subsubsection{Estimation of the Structural Profit Function}
Given a consistent estimator of $f_{\rands{\eta}}$, we can estimate the distribution of $\pi(\rand{p}_m,\rand{e})$ conditional on a given market $m$. Note that since $\rand{e}$ has finite support, the observed distribution of profits in market $m$ is a finite mixture
\[
f_{\rands{\pi}|\rand{p}}(\cdot|p)=\sum_{e=1}^{d_e}f_{\rands{\eta}}(\cdot-\pi(p,e))\Prob{\rand{e}=e|\rand{p}=p}.
\]
Hence, any finite-mixture estimators can be applied if a consistent estimator of $f_{\rands{\eta}}$, $\hat{f}_{\rands{\eta}}$, is given. In the simplest case when the number of types $d_e$ is known,\footnote{The number of types can also be estimated. See, for example, \citet{manole2021estimating} and references therein.} we can define a parametric log-likelihood 
\[
\hat{L}(\theta)=\sum_{i=1}^N \log\left(\sum_{e=1}^{d_{e}}\hat{f}_{\rands{\eta}}(\rands{\pi}_{i,m}-\pi_e)\rho_e\right),
\]
where $\theta=\left((\pi_e)\tr_{e=1,\dots,d_{e}},(\rho_e)\tr_{e=1,\dots,d_{e}}\right)\tr$ is a vector of parameters of interest, and then find a maximum-likelihood estimator (MLE) as a solution to
\begin{align*}
&\max_{\theta}\hat{L}(\theta)\\
\text{s.t. }&\sum_{e=1}^{d_{e}} \rho_{e}=1,\\
&\rho_e\geq0, \quad e=1,\dots, d_{e},\\
&\pi_e\leq \pi_{e+1}, \quad e=1,\dots, d_{e}-1.
\end{align*}

\section{Illustrative Empirical Application}\label{app: application}
In this section, we analyze the production of houses using data from \citet{epple2010new} in line with the model described in Section~\ref{subsec:housing}. 
\par
\noindent\textbf{Data.} The data contains information on new housing construction in Allegheny County in Pennsylvania. For every dwelling $i$ we have information about total revenue from selling the house $\rand{v}_i$, the price of land $\rand{p}_{l,i}$, materials per-acre $\rand{m}_i$, and the geographic location of the house that we use to identify the zip-code for each house. From the original sample constructed in \citet{epple2010new}, we exclude houses with the value per unit of land and the price of land above 55 and 7, respectively. There are 5,641 houses in our sample. Table~\ref{tab: summary data} provides summary statistics of our sample. Figure~\ref{fig: pl and v dist} displays a distribution of the price of land and the value per unit of land.
For more details on the original data, see \citet{epple2010new}.
\begin{table}[h]
\centering
\begin{threeparttable}
\centering
\caption{Summary Statistics}\label{tab: summary data}
\begin{tabular}{lccccc}
\hline
\hline
Variable                    & Mean      & Median     & Std    & Min    & Max  \\ 
Value per unit of land      & 14.43   & 13.24 & 8.6 & 0.15 & 54.02\\
Price of land               & 2.26  & 2.11 & 1.28 & 0.05 & 6.99\\
\hline
\end{tabular}
\vspace{1ex}
\begin{tablenotes}
\item {\footnotesize Notes: These summary statistics illustrate the heterogeneity of the value per unit of land and price of land.}
\end{tablenotes}
\end{threeparttable}
\end{table}

\begin{figure}	
\centering
\includegraphics[width=.6\textwidth]{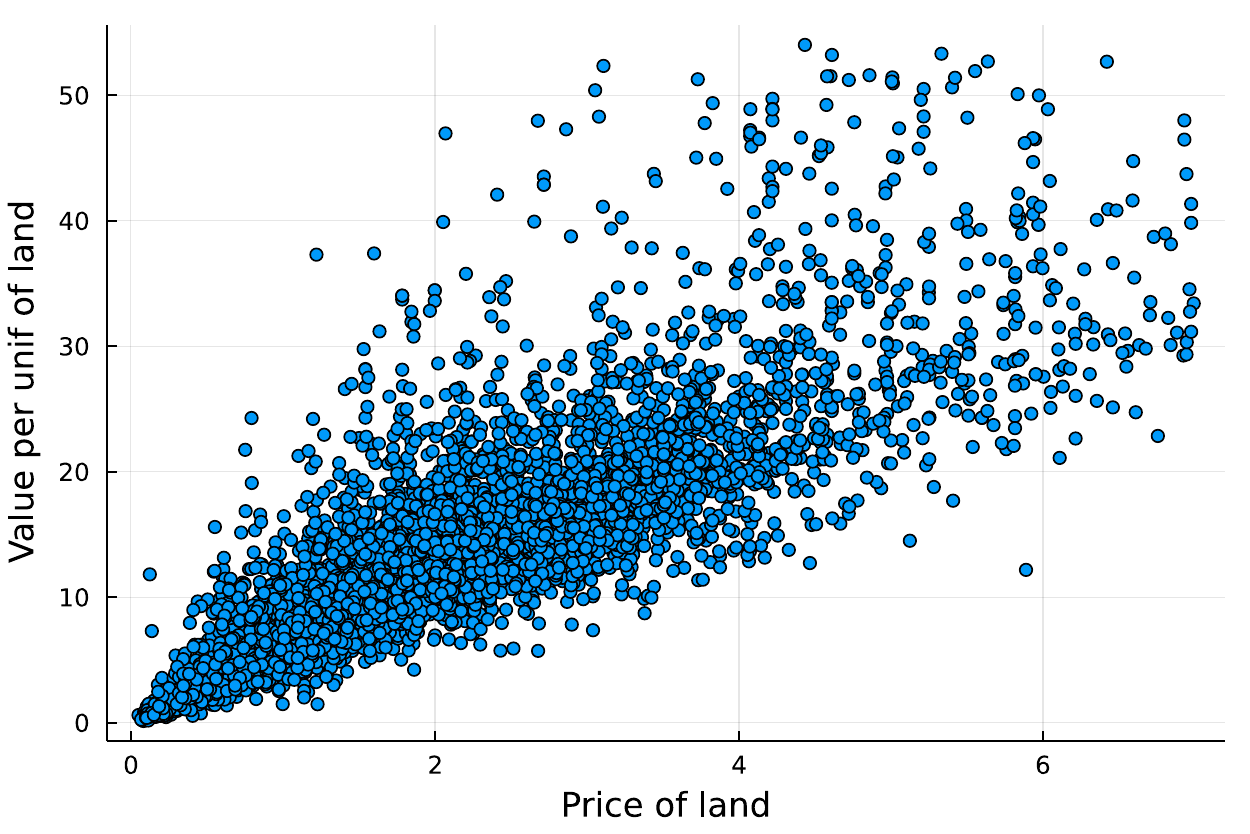}
\caption{\textsc{Distribution of Price of Land and Value per Unit of Land in Sample}}\label{fig: pl and v dist}
\end{figure}
\par
\noindent\textbf{Market definition.} We highlight that our method takes markets as known, but in practice we have to define them. We assume that within Allegheny County, local markets are determined by the location of the house (coordinates) and the price of land. To construct the markets, we use K-means clustering using location and the price of land. We select the number of cluster using the heuristic \emph{elbow} method. We end up with $8$ markets.  After clustering the observations, we average the price of land within the market to obtain the market level price of land $\rand{p}_m$. The distribution of the price of land and the value per unit of land across the local markets are depicted on Figures~\ref{fig: price of land market distribution} and~\ref{fig: value market distribution}.
\begin{figure}	
\centering
\includegraphics[width=.6\textwidth]{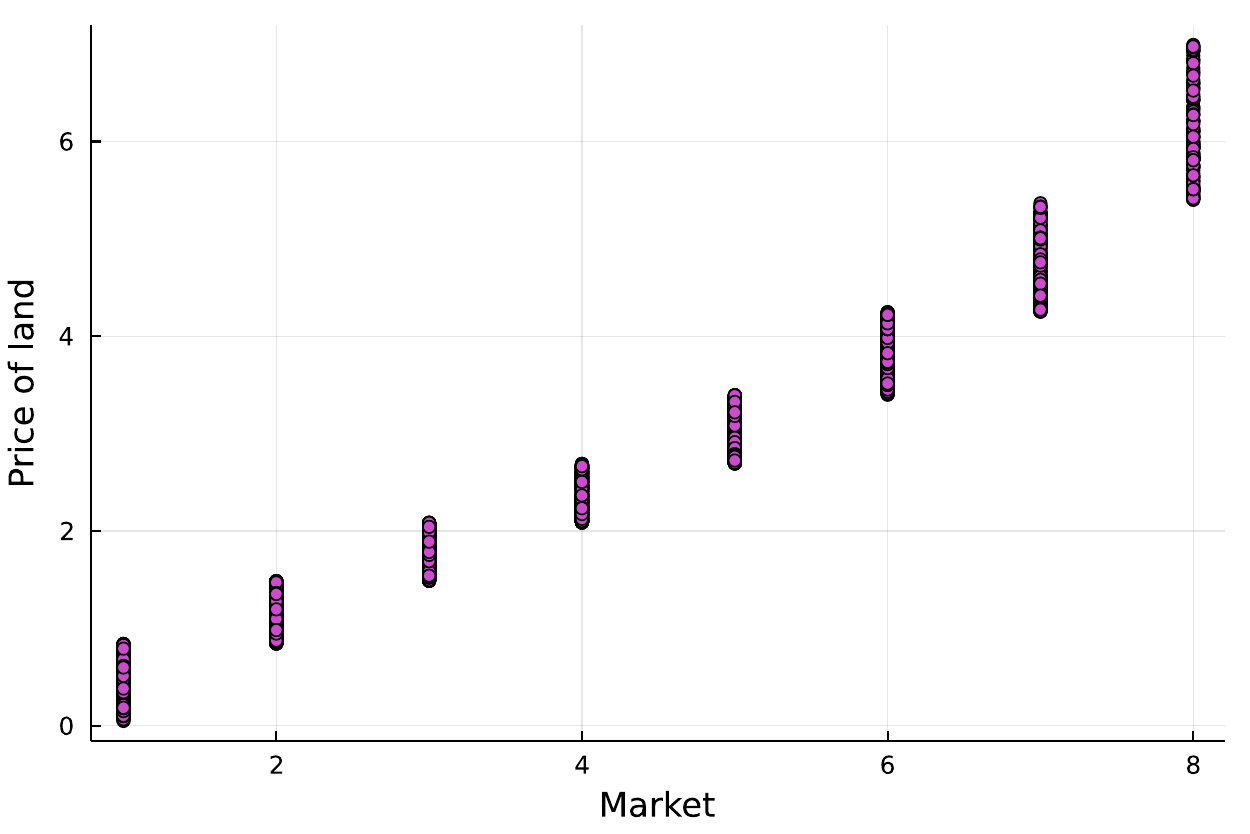}
\caption{\textsc{Distribution of Price of Land across Markets}}\label{fig: price of land market distribution}
\end{figure}
There is not much variation in the price of land  within most of the markets, but substantial variation across markets.\footnote{The largest across markets standard deviation is about $0.46$, which corresponds to market 8.} This is evidence of the validity of our assumption that firms within the same market face the same prices. At the same time, there is a lot of variation in valuation per unit of land across and within markets. Moreover, the valuations are clearly bounded from below and some markets display clear separated sets of points that line up with our assumptions of discrete heterogeneity and bounded support of measurement error.  
\begin{figure}	
\centering
\includegraphics[width=.6\textwidth]{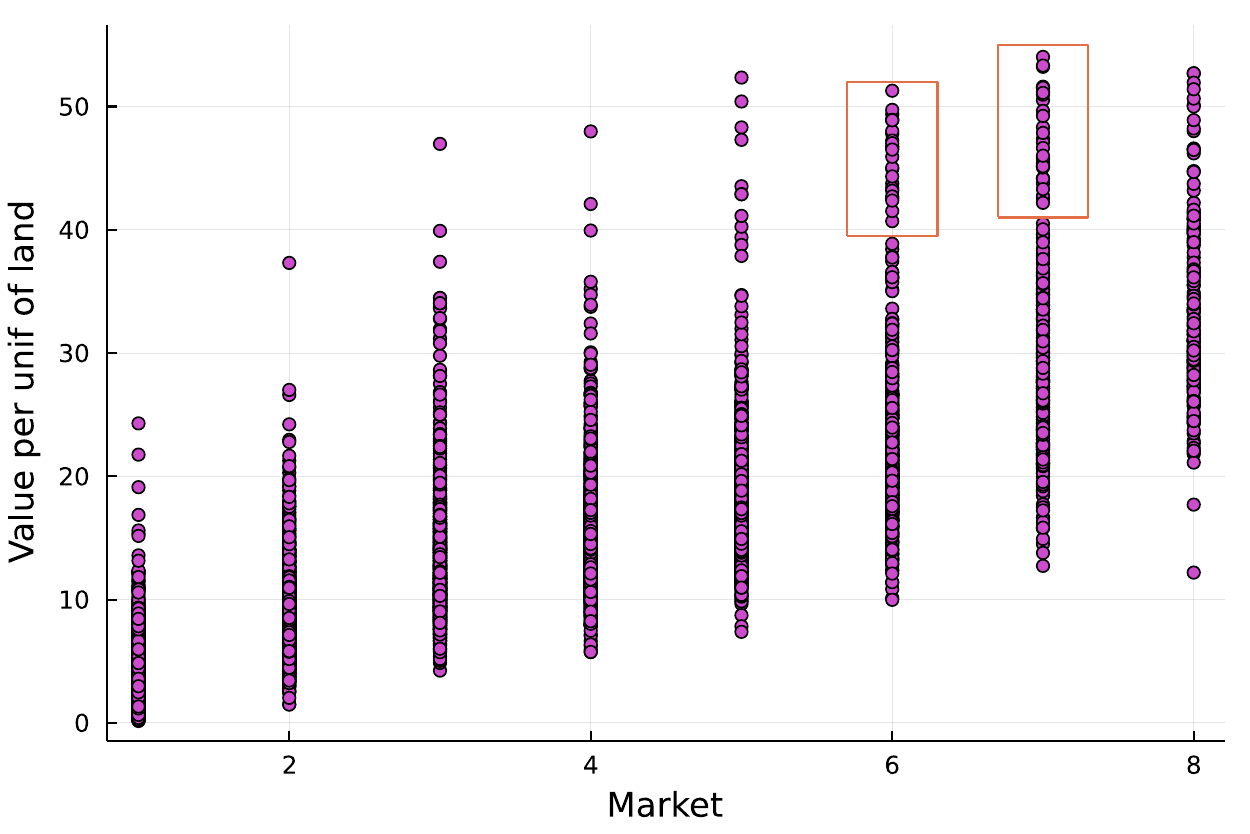}
\caption{\textsc{Distribution of Value per Unit of Land across Markets} The framed regions illustrate clear separations of the highest type for markets 6 and 7.}\label{fig: value market distribution}
\end{figure}
\par
\noindent\textbf{Estimation of the measurement error distribution.}
To estimate the distribution of the measurement error, we note that markets $6$ and $7$ exhibit two sets of observations with high value per unit of land, that are clearly separated  from the rest of observations (see Figure~\ref{fig: value market distribution}).\footnote{Formally, we use a density-based clustering technique in this step.} Moreover, these sets cover intervals of very similar length. This is consistent with the assumption that at least in one market at least one type is clearly separated from other types.\footnote{There other visible separations in Figure~\ref{fig: value market distribution}. However, those regions do not contain enough observations for nonparametric estimation of $f_{\rands{\eta}}$.} We merge observations from these two regions after recentering and then use the Epanechnikov kernel with the least-squares cross-validated bandwidth to estimate the measurement error p.d.f. 
\par
\noindent\textbf{Estimation of values per unit of land for different productivity levels.} For every market, given the estimated density of the measurement error, we apply the procedure for estimation of finite mixtures described in \citet{kim2020fast} to get a good starting point to obtain the MLE of values per unit of land for firms with different productivity $\{\hat{v}_m(e)\}_{m=1,e=1}^{M,d_e}$. We assume for simplicity that the number of types of firms is the same across markets and is equal to $d_e=4$, which is the minimal number of mixtures that is able to cover the observed support of mismeasured values per unit of land.\footnote{The results for $d_e=5$ and $d_e=6$ are qualitatively the same and available upon request.}
\par
\noindent\textbf{Proxy function.} To estimate the proxy function that maps average value per unit of land to the price of output, we follow \citet{epple2010new}. In particular, we use a 3rd order degree polynomial to estimate $\Exp{\rand{p}|\rand{\bar{v}}}$ and then solved the ordinary differential equation (\ref{eq: epple ode}). As a result, we can estimate the output price at every market $\{\hat{p}_{o,m}\}_{m=1}^M$. 
\par
\noindent\textbf{Supply.} To estimate the output level of firms with productivity $e$, we use $\{\hat{v}_{m}(e)/\hat{p}_{o,m}\}_{m=1,e=1}^{M,E}$. The resulting logarithm of supply as a function of the logarithm of the output price is depicted on Figure~\ref{fig: log supply}.
\begin{figure}
\centering
  \begin{tikzpicture}
    \begin{axis}[
    xmin=0.3, xmax=0.6,
    ymin=1.6, ymax=3.5,
    axis lines = left,
    xlabel = Logarithm of Price of output,
    ylabel = Logarithm of Ouput level,
    legend entries={$e=1$, $e=2$, $e=3$, $e=4$},
    legend style={at={(0.9,0.3)},anchor=west},
    ]
    \addplot +[smooth] coordinates {(0.3102519316493614, 1.9255372186797144)
 (0.3703529599892155, 2.1760919708407376)
 (0.3989099615318926, 2.3145943386353665)
 (0.4304505026032072, 2.3201667815870084)
 (0.47442364370006024, 2.4250161241437804)
 (0.5263274767388194, 2.460313217824752)
 (0.5516025914377112, 2.216870385750497)};
    \addplot +[smooth] coordinates {(0.3102519316493614, 1.9261593243179889)
 (0.3703529599892155, 2.8385810318206235)
 (0.3989099615318926, 2.912303202504435)
 (0.4304505026032072, 2.4490773564009385)
 (0.47442364370006024, 2.6141824660204014)
 (0.5263274767388194, 2.7016572977774973)
 (0.5516025914377112, 2.7968271644016407)};
    \addplot +[smooth] coordinates {(0.3102519316493614, 2.7450196383971033)
 (0.3703529599892155, 3.110407398117317)
 (0.3989099615318926, 3.234306080273413)
 (0.4304505026032072, 2.924264164009136)
 (0.47442364370006024, 3.0252586055021298)
 (0.5263274767388194, 2.9787199772289776)
 (0.5516025914377112, 3.0415934050209144)};
    \addplot +[smooth] coordinates {(0.3102519316493614, 3.259895639694167)
 (0.3703529599892155, 3.416281917684249)
 (0.3989099615318926, 3.4062215779461047)
 (0.4304505026032072, 3.3677515325860385)
 (0.47442364370006024, 3.353137616536492)
 (0.5263274767388194, 3.3440311460032253)
 (0.5516025914377112, 3.287030671428205)};
    \end{axis}
    \end{tikzpicture}
\caption{Supply curves for different productivity levels for 7 markets with the highest price of output.}
\label{fig: log supply}
\end{figure}
The supply curves are close to be monotonically increasing. We attribute nonmonotonicity to estimation error. Next we enforce monotonicity by finding the output levels that preserve monotonicity in the output price and minimize the Euclidean distance to the estimated output level. The resulting logarithm of supply, for firms with different productivity, as a function of the logarithm of the output price is depicted on Figure~\ref{fig: log supply2}.
\begin{figure}
\centering
  \begin{tikzpicture}
    \begin{axis}[
    xmin=0.3, xmax=0.6,
    ymin=1.6, ymax=3.5,
    axis lines = left,
    xlabel = Logarithm of Price of output,
    ylabel = Logarithm of Ouput level,
    legend entries={$e=1$, $e=2$, $e=3$, $e=4$},
    legend style={at={(0.9,0.3)},anchor=west},
    ]
    \addplot +[smooth] coordinates { (0.3102519316493614, 1.9255372186679525)
 (0.3703529599892155, 2.176091970839345)
 (0.3989099615318926, 2.314594311510357)
 (0.4304505026032072, 2.32016680855113)
 (0.47442364370006024, 2.3730264842572795)
 (0.5263274767388194, 2.373026484283812)
 (0.5516025914377112, 2.3730264842877262)};
    \addplot +[smooth] coordinates { (0.3102519316493614, 1.926159324319918)
 (0.3703529599892155, 2.716398487122548)
 (0.3989099615318926, 2.7163984871332825)
 (0.4304505026032072, 2.7163984871313525)
 (0.47442364370006024, 2.7163984872000118)
 (0.5263274767388194, 2.716398491967432)
 (0.5516025914377112, 2.7968271644054745)};
    \addplot +[smooth] coordinates { (0.3102519316493614, 2.7450196384007293)
 (0.3703529599892155, 3.057454253663)
 (0.3989099615318926, 3.057454255363498)
 (0.4304505026032072, 3.057454255509833)
 (0.47442364370006024, 3.0574542561825777)
 (0.5263274767388194, 3.057454257411642)
 (0.5516025914377112, 3.057454299799494)};
    \addplot +[smooth] coordinates { (0.3102519316493614, 3.259895639696557)
 (0.3703529599892155, 3.363316561389958)
 (0.3989099615318926, 3.3633165614069447)
 (0.4304505026032072, 3.3633165614135088)
 (0.47442364370006024, 3.3633165614202003)
 (0.5263274767388194, 3.363316561427752)
 (0.5516025914377112, 3.3633165614357368)};
    \end{axis}
    \end{tikzpicture}
\caption{Monotone supply curves for different productivity levels for 7 markets with the highest price of output.}
\label{fig: log supply2}
\end{figure}
\par
\noindent\textbf{Discussion.} Our results indicate that there is substantial heterogeneity in the supply of housing. Recall that the results in \citet{epple2010new} focus on a representative firm. In contrast, our results suggest that we cannot ignore heterogeneity.   For instance, factor-reallocation is total productivity enhancing when resources shift from the less productive to the more productive firms \citep{melitz2014heterogeneous}. 
\par
In the production of housing, heterogeneity can be interpreted as \emph{curb appeal} \citep{epple2010new}. The most productive firms ($e=4$) produce the houses with the highest curb appeal. In that sense, it is important to disentangle this heterogeneity when estimating the housing supply elasticity. Housing supply elasticity has been seen as a key parameter \citep{glaeser2005have} to understand the relationship between urban growth and new residential construction. An inelastic supply means that a positive regional shock will lead to higher paid workers and more expensive  houses. If the housing supply is elastic, then we can expect smaller price changes and expansion of the size of the city. Using Figure~\ref{fig: log supply2}, we computed the average elasticity for all types of firms (Table~\ref{tab: elasticities}). 

\begin{table}[h]
\centering
\begin{threeparttable}
\centering
\caption{Average Elasticities of Firms with Different Productivity}\label{tab: elasticities}
\begin{tabular}{lcccc}
\hline
\hline
                        & $e=1$ & $e=2$ & $e=3$ & $e=4$\\ 
Average Elasticity      & 1.73  & 2.72  & 0.87  & 0.29 \\
\hline
\end{tabular}
\vspace{1ex}
\end{threeparttable}
\end{table}
We observe from Table~\ref{tab: elasticities} that the most productive firm type ($e=4$) has a very inelastic supply. This means that houses with the highest curb appeal will mainly see an increase of prices (without a large expansion) due to a positive regional shock. In contrast, we observe that the least productive firm types ($e=1$ and $e=2$) have elastic supplies. This means that there will be an expansion in the construction (with an smaller increase of prices) of houses with the lowest curb appeal as a result of the same positive regional shock. 
\end{document}